\documentclass[journal]{IEEEtran}

\usepackage{ifpdf}

\usepackage{cite}

\ifCLASSINFOpdf

\else

\fi

\usepackage{amsmath}
\usepackage{gensymb}

\usepackage{algorithmic}

\usepackage{array}
\usepackage{multirow}

\usepackage{url}
\usepackage{xcolor}

\usepackage{cite}
\usepackage{amsmath, amssymb, amsfonts}
\usepackage{algorithmic}
\usepackage{booktabs}
\usepackage{graphicx}
\usepackage{textcomp}
\usepackage{bm}
\usepackage{amsthm}

\usepackage[
    colorlinks=true,
    linkcolor=cyan,
    urlcolor=cyan,
    citecolor=cyan,
    anchorcolor=cyan
]{hyperref}

\theoremstyle{definition}
\newtheorem{definition}{Definition}
\newtheorem{theorem}{Theorem}

\newtheorem{remark}{Remark}
\newtheorem{assumption}{Assumption}

\begin{document}
    
    \title{When Can Phasor-Domain Device Models Be Trusted for Electromechanical Stability Analysis of Grid-Forming Converter-Dominated Microgrids?}

    \author{Zhongze~Li, Xiaoyu~Peng,
    Xi~Ru,~Zhaojian~Wang,~Jianxin~Zhang,~Yingshang~Liu,~and~Feng~Liu

    \thanks{This paper is supported by the Smart Grid National Science and Technology Major Project under Grant 2024ZD0801200 and the Science and Technology Project of China Southern Power Grid Co., Ltd under Grant GDKJXM20250201(036000KC25020005). (\textit{Corresponding Author: Feng Liu}).}
\thanks{Zhongze Li, Xiaoyu Peng, Xi Ru, and Feng Liu are with Department of Electrical Engineering, Tsinghua University, Beijing, China. Feng Liu is also with the State Key Laboratory of
Power System Operation and Control, Tsinghua University,
China (email: lizz24@mails.tsinghua.edu.cn;
pengxy19@tsinghua.org.cn; 
thu.ruxi@gmail.com; 
 lfeng@tsinghua.edu.cn).}
\thanks{Zhaojian Wang is with School of Automation and Intelligent Sensing, Shanghai Jiao Tong University, China (email: wangzhaojian@sjtu.edu.cn).}
\thanks{Jianxin Zhang and Yingshang Liu are with Power Dispatching Control Center, China Southern Power Grid, China.} 

    }

    \maketitle

    \begin{abstract}
    Grid-forming (GFM) converter-dominated microgrids are frequently analyzed using reduced-order phasor-domain electromechanical GFM models, but the validity of these models is often taken for granted. 
    Assuming ideal inner-loop tracking (IILT) of terminal-voltage references, these models neglect the inner-loop and filter dynamics at the electromagnetic-transient (EMT) timescale to simplify stability analysis.
    This paper argues that such neglected dynamics can destabilize the system, invalidating the stability conclusions drawn from the IILT model.
    To address this cross-timescale stability issue, we formulate the validity of the IILT stability conclusion as a robust-stability certification problem. 
    The EMT-induced model mismatch between the reduced-order converter model and the actual converter model
    is represented as a structured uncertainty embedded around the IILT feedback loop. This yields a frequency-resolved interaction index and a sufficient certificate based on the structured singular value for determining when the stability conclusion of the IILT model can be certified under a prescribed EMT uncertainty weight. The uncertainty weight can be obtained from detailed EMT models or terminal reference-response measurements. Case studies show that the certificate verifies small-signal model validity when its sufficient condition holds and identifies operating points that require further EMT assessment when certification is lost. The measurement-based weights also closely match their model-based counterparts without requiring access to inner-loop models.
    \end{abstract}

 
    \IEEEpeerreviewmaketitle
    \bstctlcite{BSTcontrol}

    \section{Introduction}

    Grid-forming (GFM) converters are increasingly deployed in converter-dominated microgrids. 
    In many electromechanical-timescale stability studies, they are represented by reduced-order phasor-domain device models that retain outer power-synchronization and voltage-regulation dynamics while assuming ideal or sufficiently fast reference-to-terminal-voltage dynamics
    \cite{SCHIFFER2016135,fuQoAdmittanceModeling2023a,pengCompositionalGridCode2025,xiuqiang}.
    These models are compatible with network models in polar coordinates, ranging from quasi-static power-flow models to dynamic line models  \cite{fuQoAdmittanceModeling2023a,pengCompositionalGridCode2025}. However, a key question is whether such models can be trusted for electromechanical stability analysis in multi-GFM microgrids. Although the target phenomenon occurs on the electromechanical timescale, the terminal voltage of a GFM converter is realized through electromagnetic-transient (EMT)-scale inner loops, filters, feedforward paths, and network-dependent terminal interactions. Consequently, the neglected reference-to-terminal-voltage tracking dynamics
    may interact with the multi-converter network, invalidating the stability assessment obtained from the reduced-order model \cite{ModelTF}.
    
Existing studies have shown that converter-side EMT dynamics can reshape low-frequency stability. Inner voltage and current loops may introduce non-negligible phase lag and coupling effects \cite{zhaoClosedFormSolutionsGridForming2024,ravanjiImpactVoltageLoopFeedforward2023}. Low-frequency resonance mechanisms have been reported in GFM systems \cite{zhaoLowFrequencyResonancesGridForming2024}. Angle--voltage coupling and power-dynamic decoupling studies further show that terminal voltage magnitude and phase channels cannot always be treated independently \cite{pengImpactofAngle-VoltageCoupling,zhaoPowerDynamicDecoupling2022}. Impedance-based, circuit-based, and complex-torque-based methods provide complementary views of such interactions \cite{liImpedanceCircuitModel2021,zhaoExploringDampingEffecta}. These studies reveal important 
cross-timescale mechanisms by which converter-side EMT dynamics affect low-frequency angle--voltage stability, but they do not directly certify when a reduced-order phasor-domain device model preserves the correct stability assessment for multi-converter microgrids.

Model reduction and timescale separation methods provide an alternative approach. Singular-perturbation theory gives a principled reduction when fast and slow dynamics are sufficiently separated \cite{christofidesSingularPerturbationsInputtostate1996,maSingularPerturbationBasedLargeSignal2024}. Related reductions have been developed for microgrids and converter-dominated systems \cite{zhaozhuoliModel_reduction,rasheduzzaman_sp_in_mg}. Reduced-order modeling and transfer-function approximation further improve tractability for larger systems \cite{ajalaModelReductionDynamic2023,eberleinImpactInnerControl2023}. 
For example, \cite{MISYRIS2021107324} uses bandwidth-separation and stability-margin conditions to determine when root-mean-square (RMS) models are adequate for active-power studies of grid-connected GFM converters.
However, these approaches often depend on detailed device models, operating conditions, controller parameters, network strength, or identifiable dominant states.
When timescale separation is weak or varies across heterogeneous devices, a phasor-domain model lacks an explicit certificate that its stability conclusion transfers to the corresponding EMT system under the neglected converter-side dynamics.

Robust control and measurement-based methods are relevant but do not fully close this gap. Structured singular-value analysis provides a standard framework for robust stability under structured uncertainty \cite{zhou1996robust}. 
In converter-dominated systems, robust stability tools have been used to assess uncertainty in converter interactions and control loops \cite{RossoGFMs, JinSSV}. 
In particular, \cite{Almawu25} constructs a grey-box GFM model with loop-wise multiplicative uncertainties to assess robust stability across possible controller implementations.
However, these tools either require explicit state-space uncertainty descriptions and parameter bounds \cite{JinSSV} or rely on heuristically chosen uncertainty weights for specific control loops \cite{Almawu25,RossoGFMs}.
Robust and adaptive control studies provide valuable complementary perspectives, including decentralized control designs that remain robust under malicious attacks \cite{PWang}.
However, distinct from prior studies, we formulate
the EMT reduction error relative to a prescribed ideal inner-loop tracking (IILT) model as structured
uncertainty, enabling model-validity certification.

Measurement-based and black-box methods can extract terminal dynamic responses from perturbation data \cite{haberleOptimalDynamicAncillary2025,CifuentesBlack-Box}. Nevertheless, measured frequency responses alone do not provide a system-level validity certificate for a phasor-domain device model. 
What is still missing is a bridge from measurement-based EMT information to a robustness-compatible uncertainty description that can certify whether the stability conclusion of the IILT model remains valid.

To address these issues, this paper proposes an assessment method based on robust $\mu$-analysis theory to certify the trustworthiness of phasor-domain device models for stability analysis. 
The main contributions are twofold:
\begin{itemize}
    \item \textbf{Theory}: We formulate the trustworthiness of phasor-domain device models as a robust model-validity problem for GFM-dominated microgrids. 
    Unlike existing studies that directly analyze detailed EMT models or compare reduced-order and EMT models on a case-by-case basis, the proposed method introduces the IILT model as the nominal device model. It embeds the EMT-induced reference-to-terminal-voltage mismatch as a structured uncertainty in the closed loop. 
    Consequently, the model-validity problem is converted into a robust stability problem. 
    This conversion yields a frequency-resolved interaction index and a structured singular-value certificate
    that justify when the phasor-domain device model can be trusted for stability analysis.
    Compared with an unstructured small-gain bound, the block-diagonal uncertainty structure reduces conservatism.

    \item \textbf{Practical Implementation}: We develop model-based and measurement-based procedures for constructing the per-device EMT-induced uncertainty weight $\mathbf{W}(s)$ required by the certificate. 
    The construction captures angle-channel, voltage-channel, and cross-channel mismatch between voltage phasor references and terminal voltage phasors, rather than relying on scalar bandwidth or strict timescale separation. 
    It therefore connects detailed EMT models or terminal reference-response measurements to a system-level model validity certificate, enabling validity assessment when converter inner-loop structures are heterogeneous, proprietary, or unavailable.
\end{itemize}

    The rest of this paper is organized as follows. Section \ref{sec: MicrogridModel} develops the microgrid model. Section \ref{sec: EMTFeedback} reformulates EMT-induced mismatch as structured uncertainty around the IILT closed loop. Section \ref{sec: RobustStability} defines the interaction index and derives a sufficient robust-stability certificate. Section \ref{sec: PracticalImplementation} presents model-based and measurement-based procedures for constructing the uncertainty weight $\mathbf{W}(s)$. Section \ref{sec: Case} reports case studies, and Section \ref{sec: Conclusion} concludes the paper.

    \textbf{Notation}: 
Throughout this paper, $\odot$ denotes the entrywise (Hadamard) product of conformable matrices, and $\mathrm{blkdiag}(\cdot)$ denotes a block-diagonal matrix. The $\mathcal{H}_\infty$ norm of a transfer matrix is $\|\cdot\|_\infty$, and $\mathcal{RH}_\infty$ denotes the space of real-rational, stable, proper transfer functions. For a given uncertainty set $\widetilde{\mathbf{\Delta}}_e$, the structured singular value is $\mu_{\widetilde{\mathbf{\Delta}}_e}(\cdot)$. $\mathbf{I}_k$ denotes the $k$-dimensional identity matrix.
A hat ( $\hat{\cdot}$ ) denotes complex circuit variables.
A tilde ( $\widetilde{\cdot}$ ) denotes the lifted form of the corresponding entrywise quantity, and asterisk ($\cdot^*$) denotes the value at the equilibrium.

    \section{Converter-Dominated Microgrid Modeling}\label{sec: MicrogridModel}
    This section develops an EMT model of a microgrid with
    $n$ GFM converters.
    We start by modeling each
    converter and assembling them into the overall interconnected feedback model.

    \subsection{GFM Modeling}

    \subsubsection{Outer Loop Dynamics}
    For the $i$-th converter, the outer loops generate the reference angle $\delta^{\mathrm{ref}}_{i}$, frequency $\omega
    ^{\mathrm{ref}}_{i}$, and voltage $V^{\mathrm{ref}}_{i}$ via virtual synchronous generator (VSG) control.
    \begin{align}
        \label{eq: outer angle loops}s \delta^{\mathrm{ref}}_{i}=\omega_{0}\omega_i^{\mathrm{ref}},s\omega^{\mathrm{ref}}_{i} & = \frac{1}{M_{i}}\left(- D_{i}(\omega^{\mathrm{ref}}_{i}-1) + P^{\mathrm{ref}}_{i}- P_{i}\right),
    \end{align}
    \begin{equation}
        V^{\mathrm{ref}}_{i}= V_{0,i}-D_{v,i}(Q_{i}-Q^{\mathrm{ref}}_{i}).\label{eq: outer voltage loops}
    \end{equation}
    Here, $M_{i}$, $D_{i}$, and $D_{v,i}$ denote the virtual inertia, frequency
    damping, and reactive-power droop gain, respectively. $\omega
    _{0}$ is the nominal frequency.
    $P^{\mathrm{ref}}_{i}$, $Q^{\mathrm{ref}}_{i}$,
    and $V_{0,i}$ are the active/reactive power and voltage setpoints. 
    The proposed method can be extended to GFMs with more complex outer loops, such as droop and
    dispatchable virtual oscillator control (dVOC)~\cite{xiuqiang}. Here, VSG control is adopted for illustration.

    The output voltage and current are $\hat{V}_{\text{o} dq,i}=V_{\text{o}d,i}+jV_{\text{o}q,i}$ and $\hat{I}_{\text{o} dq,i}={I}_{\text{o} d,i}+j {I}_{\text{o} q,i}$, respectively, yielding
    \begin{equation}
        P_{i}= \frac32\left( V_{\text{o}d,i}I_{\text{o}d,i}+ V_{\text{o}q,i}I_{\text{o}q,i}
    \right),
        Q_{i}= \frac32\left( V_{\text{o}q,i}I_{\text{o}d,i}- V_{\text{o}d,i}I_{\text{o}q,i}
    \right).
    \end{equation}

    \subsubsection{Inner Loop and Filter Dynamics}
    The detailed converter model adopted in this paper is a nonlinear,
continuous-time averaged-switch model. Its filter dynamics
are given by
    \begin{align}
        L_{\text{f},i}s \hat{I}_{\text{L} dq,i} & = \hat{V}_{\text{i} dq,i}- \hat{V}_{\text{o} dq,i}- j\omega_{\text{r},i}L_{\text{f},i}\hat{I}_{\text{L} dq,i}- r_{\text{f},i}\hat{I}_{\text{L} dq,i},\label{eq:Lf_current} \\
        C_{\text{f},i}s \hat{V}_{\text{o} dq,i} & = \hat{I}_{\text{L} dq,i}- \hat{I}_{\text{o} dq,i}- j\omega_{\text{r},i}C_{\text{f},i}\hat{V}_{\text{o} dq,i},\label{eq:Cf_voltage}
    \end{align}
    where $\hat{V}_{\text{i} dq,i}$ is the converter internal voltage and
    $\hat{I}_{\text{L} dq,i}$ is the filter inductor current.
    $\omega_{\text{r},i}= \omega_{0}\omega^{\mathrm{ref}}_{i}$ is the angular frequency
    of the $dq$ frame generated by the angle loop.

    The inner voltage and current control loops, which employ proportional-integral (PI) controllers, are modeled as
    \begin{equation}
        \begin{aligned}
        \hat{I}_{\text{L} dq,i}^{\mathrm{ref}}= & \left(k_{\text{vp},i}+{k_{\text{vi},i}}/{s}\right) \left(\hat{V}_{\text{o}dq,i}^{\mathrm{ref}}-\hat{V}_{\text{o}dq,i}\right) \\
        & + \hat{I}_{\text{o}dq,i}+j \omega_{\text{0}}C_{\text{f},i}\hat{V}_{\text{o} dq,i},\label{eq:Voltage_loop}
        \end{aligned}
        \end{equation}
        \begin{equation}
        \begin{aligned}
        \hat{V}_{\text{i} dq,i}=              & \left(k_{\text{ip},i}+{k_{\text{ii},i}}/{s}\right) \left(\hat{I}_{\text{L}dq,i}^{\mathrm{ref}}-\hat{I}_{\text{L}dq,i}\right) \\
        & +\hat{V}_{\text{o} dq,i}+j \omega_{\text{0}}L_{\text{f},i}\hat{I}_{\text{L}dq,i},\label{eq:Current_loop}
        \end{aligned}
    \end{equation}
    where $\hat{V}_{odq,i}^{\text{ref}} = V_i^{\text{ref}} + j0$.
    The overall converter model is adopted from \cite{liImpedanceCircuitModel2021}.

    \subsection{Coordinate Transformation}
    To facilitate closed-loop stability analysis, the local $dq$ coordinate of each converter is transformed into the global $DQ$ frame rotating at the angular frequency $\omega_0$.
    In the $DQ$ reference frame, the $i$-th bus is represented by its complex voltage $\hat{V}_{\text{o} DQ, i}= V_{i}\angle \theta_{i}$, where $V_{i}$ and $\theta_{i}$ are the magnitude and phase angle of the terminal voltage, respectively. Define the angle of the local $d$-axis in the global $DQ$ frame as
    $\theta^{\mathrm{ref}}
    _{i}$,
    satisfying $\theta_i^{\mathrm{ref}} = \delta_i^{\mathrm{ref}}-\omega_0t$. Denote the rotation matrix by
    \begin{equation}
        \mathcal{T}(\theta) \triangleq
        \begin{bmatrix}
            \cos \theta & -\sin \theta \\
            \sin \theta & \cos \theta
        \end{bmatrix}
        .\label{eq:T_theta}
    \end{equation}
    Then, the coordinate transformation between $dq$ and $DQ$ is
    \begin{equation}
        \begin{aligned}
            \begin{bmatrix}I_{\text{o}d,i}\\ I_{\text{o}q,i}\end{bmatrix} & = \mathcal{T}^{-1}(\theta_{i}^{\mathrm{ref}}) \begin{bmatrix}I_{\text{o}D,i}\\ I_{\text{o}Q,i}\end{bmatrix},\
            \begin{bmatrix}V_{\text{o}D,i}\\ V_{\text{o}Q,i}\end{bmatrix} = \mathcal{T}(\theta_{i}^{\mathrm{ref}})\begin{bmatrix}V_{\text{o}d,i}\\ V_{\text{o}q,i}\end{bmatrix},
        \end{aligned}
        \label{eq:dq_to_DQ}
    \end{equation}
    where the subscripts $d, q$ and $D, Q$ denote the components in the local $dq$ and global $DQ$ frames, respectively. 

    \subsection{Network Modeling}
    In the global $DQ$ reference frame, the transmission line dynamics between buses $i$
    and $j$ can be written as
    \begin{equation}
        sL_{\mathrm{g},ij}\hat{I}_{\text{o},ij} = \hat{V}_{\text{o} DQ,i}- \hat{V}_{\text{o} DQ,j}- j\omega_{\text{0}}
        L_{\text{g},ij}\hat{I}_{\text{o},ij}- r_{\text{g},ij}\hat{I}_{\text{o},ij}
        ,\label{eq:Lg_current}
    \end{equation}
    where $L_{\text{g},ij},r_{\text{g},ij}$ are the line inductance and resistance, $\hat{I}_{\text{o},ij}= I_{\text{o}D,ij}+j I_{\text{o}Q,ij}$ is the
    line current, and $\omega_{\text{0}}$ is the frequency of the global
    $DQ$ frame. 

    At bus $i$, the dynamics of the shunt inductive and capacitive branches can be described respectively by

\begin{align}
sL_{g,ii}\hat{I}_{L,ii} &= \hat{V}_{oDQ,i} - j\omega_0 L_{g,ii}\hat{I}_{L,ii} - r_{g,ii}\hat{I}_{L,ii}, \\
sC_{g,ii}\hat{V}_{oDQ,i} &= \hat{I}_{C,ii} - j\omega_0 C_{g,ii}\hat{V}_{oDQ,i} - g_{g,ii}\hat{V}_{oDQ,i},
\end{align}

where $L_{g,ii}$, $r_{g,ii}$ are the shunt inductance and resistance, $C_{g,ii}$, $g_{g,ii}$ are the shunt capacitance and conductance, and $\hat{I}_{o,ii} = \hat{I}_{L,ii} + \hat{I}_{C,ii} = I_{oD,ii} + jI_{oQ,ii}$ is the total shunt current in the $DQ$ frame.

    \subsection{Interconnected Feedback Model}
    The overall multi-converter microgrid is modeled as a feedback interconnection between GFM converters and network dynamics. Specifically, the network maps terminal voltages $\hat{V}_{\text{o} DQ,i}$ to injected currents $\hat{I}
    _{\text{o} DQ,i}$, whereas each GFM maps measured currents
    $\hat{I}_{\text{o} DQ,i}$ to terminal voltages $\hat{V}_{\text{o} DQ,i}$. 
    The relationships among different coordinates are shown in Fig.~\ref{fig:mu analysis}(a).
    \begin{figure}[h]
        \centering
        \includegraphics[width=\columnwidth]{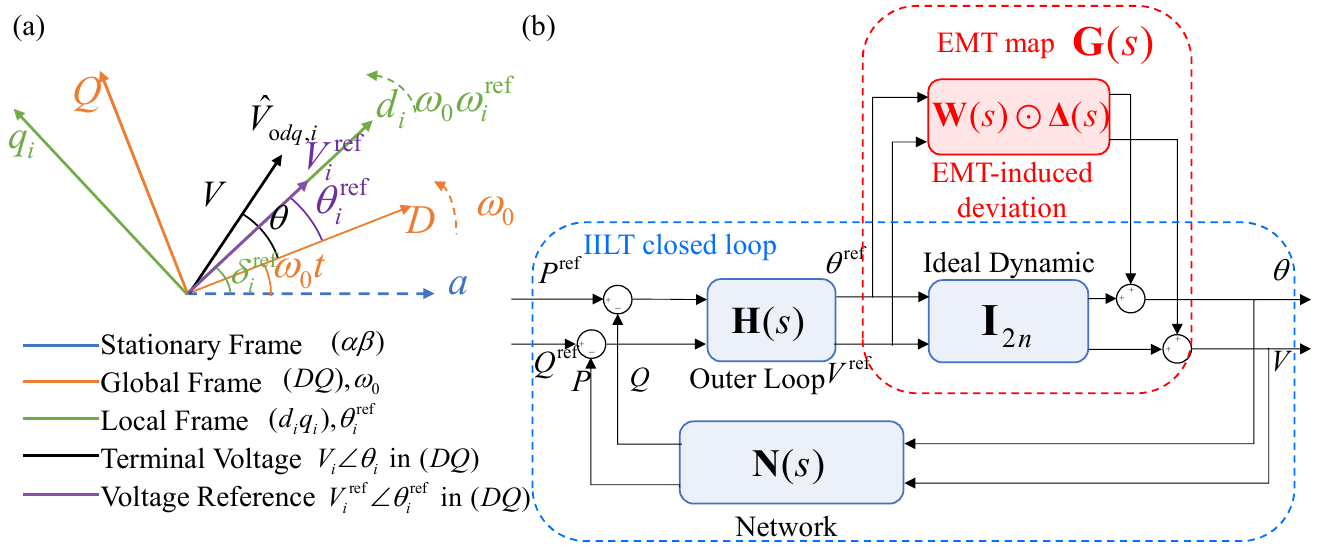}
        \caption{(a) The relationships among $DQ$, $dq$ and $\alpha\beta$ coordinates, and the vector representation of reference and terminal voltages. (b) Block diagram of the multi-GFM microgrid with converter-side EMT dynamics reformulated as an uncertain interconnection structure. The EMT-induced mismatch $\mathbf{W}(s)\odot\mathbf{\Delta}(s)$ enters the IILT feedback loop.
}
        \label{fig:mu analysis}
    \end{figure}

    \section{Formulating EMT Impacts as Uncertainty}
    \label{sec: EMTFeedback}

    In this section, we formulate the device model as the interconnection between the IILT dynamics and the fast inner-loop EMT dynamics. Since the exact EMT impacts are usually unknown, we model them as a set of uncertain dynamics. Then the original model-validity issue is converted into a tractable robust-stability analysis problem.

    \subsection{IILT Dynamics Modeling}

    For the subsequent small-signal analysis, the nonlinear averaged-switch model is linearized around the considered equilibrium.
    First, we formulate the IILT dynamics
    by linearizing the system at
    the equilibrium and expressing voltage perturbations in polar
    coordinates with $\hat{V}_{\text{o} DQ,i}= V_{i}\angle \theta_{i}$, which assumes $V_{\mathrm{o}q,i}^{*}=0$ and $V_{\mathrm{o}d,i}^{*}=V_i^{*}$. By linearizing \eqref{eq:dq_to_DQ}, we obtain $\Delta V_{\mathrm{o} d,i}= \Delta V_{i}$ and
    $\Delta V_{\text{o}q,i}= V_{\text{o}d,i}^{*}(\Delta\theta_{i}- \Delta\theta_{i}
    ^{\mathrm{ref}})$, where the superscript $*$ represents the equilibrium value. 
    The IILT model retains the phase and voltage dynamics, which can be modeled as a feedback interconnection of outer-loop dynamics
    $\mathbf{H}_{i}(s)$ and network dynamics $\mathbf{N}(s)$. For simplicity, we further define the following vectors:
    \begin{align}
        \Delta\mathbf{x}                & \triangleq[\Delta\theta_{1},\,\Delta V_{1},\ldots,\,\Delta\theta_{n},\,\Delta V_{n}]^{\top}\in\mathbb{R}^{2n},\nonumber                            \\
        \Delta\mathbf{s}                & \triangleq[\Delta P_{1},\,\Delta Q_{1},\ldots,\,\Delta P_{n},\,\Delta Q_{n}]^{\top}\in\mathbb{R}^{2n},\nonumber                                    \\
        \Delta\mathbf{s}^{\mathrm{ref}} & \triangleq[\Delta P_{1}^{\mathrm{ref}},\,\Delta Q_{1}^{\mathrm{ref}},\,\ldots,\,\Delta P_{n}^{\mathrm{ref}},\,\Delta Q_{n}^{\mathrm{ref}}]^{\top}\in\mathbb{R}^{2n}.
    \end{align}

    \subsubsection{IILT Dynamics of Converters}
    For the $i$-th converter,
    the converter terminal voltage, reference input, and power injections are
\begin{equation}
    \Delta\mathbf{x}_i \triangleq
    \begin{bmatrix}
        \Delta \theta_i \\
        \Delta V_i
    \end{bmatrix},
    \;
    \Delta\mathbf{x}_i^{\mathrm{ref}} \triangleq
    \begin{bmatrix}
        \Delta \theta_i^{\mathrm{ref}} \\
        \Delta V_i^{\mathrm{ref}}
    \end{bmatrix},
    \;
    \Delta\mathbf{s}_i \triangleq
    \begin{bmatrix}
        \Delta P_i \\
        \Delta Q_i
    \end{bmatrix}.\label{eq: x_is_i}
\end{equation}
Here, $\Delta\mathbf{s}_i$ 
denotes the power-flow disturbance.
    The outer-loop transfer matrix
    $\mathbf{H}_{i}(s)\in\mathbb{C}^{2\times 2}$ maps the local active/reactive power deviation to the reference disturbance as
    \begin{equation}
        \label{eq:H_interface}
        \begin{bmatrix}
            \Delta\theta^{\mathrm{ref}}_{i} \\
            \Delta V^{\mathrm{ref}}_{i}
        \end{bmatrix}
        =\mathbf{H}_{i}(s)
        \begin{bmatrix}
            \Delta P^{\mathrm{ref}}_{i}-\Delta P_{i} \\
            \Delta Q^{\mathrm{ref}}_{i}-\Delta Q_{i}
        \end{bmatrix}.
    \end{equation}
   Equation \eqref{eq:H_interface} provides a unified description of heterogeneous converter dynamics. For VSG outer-loop dynamics
    \eqref{eq: outer angle loops}--\eqref{eq: outer voltage loops}, $\mathbf{H}_{i}(s)$ can be written as
    \begin{equation}\label{eq: H(s)-vsg}
        \mathbf{H}_{i}(s)=\text{diag}([\omega_0/(M_i s^2+D_i s), D_{v,i} ]).
    \end{equation} 
    Similarly, by defining different $\mathbf{H}_{i}(s)$, the framework readily extends to other GFM outer-loop dynamics.
    For the system model, define a block-diagonal matrix $\mathbf{H}
    (s)=\mathrm{blkdiag}(\mathbf{H}_{1}(s),\ldots,\mathbf{H}_{n}(s))\in \mathbb{C}
    ^{2n\times 2n}$.

    \subsubsection{Network Dynamics}

    The linearized network map $\mathbf{N}(s): \Delta \mathbf{x} \xrightarrow{} \Delta \mathbf{s} $ relates terminal angle-voltage perturbations to power injections at $(\bm\theta^{*},\mathbf{V}
    ^{*})$. In general, a dynamic network map $\mathbf{N}(s) = \mathbf{G}_{ \mathrm{net}}(s)$ is adopted by default to capture the interaction between the high-$R/X$ line and device dynamics \cite{Niko_ReduceOrderLine,Vorobev_High-Fidelity}. $\mathbf{G}_{ \mathrm{net}}(s)$ can be calculated using the method proposed in \cite{fuQoAdmittanceModeling2023a}. When the $R/X$ ratio is small, it reduces to a static line model, i.e., $\mathbf{N}(s) = \mathbf{G}_{\mathrm{net}}(0)=\mathbf{A}_{ \mathrm{net}}$. In this case, we have 
    \begin{equation}        \label{eq:Anet_def}\Delta\mathbf{s}=\mathbf{A}_{ \mathrm{net}}\,\Delta\mathbf{x}
        ;\; \mathbf{A}_{ \mathrm{net}}\triangleq\frac{\partial\,[P_{1},\,Q_{1},\ldots,\,P_{n},\,Q_{n}]^{\top}}{\partial\,[\theta_{1},\,V_{1},\ldots,\,\theta_{n},\,V_{n}]^{\top}}
        \Big|_{\mathbf{x}^{*}},
    \end{equation}
    where $\mathbf{A}_{ \mathrm{net}}$ is the power-flow Jacobian of the Kron-reduced power network for systems with constant-impedance loads.

    \subsubsection{Closed-Loop System IILT Dynamics}
    Combining \eqref{eq:H_interface} and $\mathbf{N}(s)$ and assuming $\Delta\mathbf{x}=\Delta\mathbf{x}^{\mathrm{ref}}$, the corresponding IILT model neglecting inner-loop dynamics
    is obtained as
    \begin{equation}
        \label{eq: MET MIMO System}\Delta\mathbf{x}=\Delta\mathbf{x}^{\mathrm{ref}}=\mathbf{H}(s)\left(\Delta\mathbf{s}
        ^{\mathrm{ref}}-\mathbf{N}(s)\,\Delta\mathbf{x}\right).
    \end{equation}

    \subsection{Converter EMT Dynamics}\label{sec: III-B}
    This subsection isolates the converter EMT dynamics through a local reference-tracking map and defines the model mismatch from ideal tracking for robust-stability analysis.

    For each converter, denote by $\mathbf{G}_{\mathrm{in},i}(s)\in\mathbb{C}^{2\times 2}$ the local
    closed-loop map from the terminal reference value disturbance
    $[\,\Delta \theta^{\mathrm{ref}}_{i},\,\Delta V^{\mathrm{ref}}_{i}\,]^{\top}$
    to the terminal voltage disturbance in polar coordinates
    $[\,\Delta \theta_{i},\,\Delta V_{i}\,]^{\top}$, i.e.,
    \begin{equation}
        \begin{bmatrix}
            \Delta \theta_{i} \\
            \Delta V_{i}
        \end{bmatrix}
        = \mathbf{G}_{\mathrm{in},i}(s)
        \begin{bmatrix}
            \Delta \theta^{\mathrm{ref}}_{i} \\
            \Delta V^{\mathrm{ref}}_{i}
        \end{bmatrix}. \label{eq:Gin_def}
    \end{equation}
    Physically, $\mathbf{G}_{\mathrm{in},i}(s)$ describes how the fast converter-side
    EMT dynamics \eqref{eq:Lf_current}--\eqref{eq:Current_loop}
    shape the tracking of the electromechanical reference commands at the $i$-th bus.

    Under the IILT approximation, the EMT dynamics are neglected, and the terminal variables are assumed to follow the outer-loop references ideally, i.e., the inner-loop map from $ \Delta \mathbf{x}_i^\mathrm{ref}$ to $\Delta \mathbf{x}_i$ is $\mathbf{I}_{2}$.
    In contrast, EMT dynamics introduce an undesirable lag in reference following.
    Therefore, to quantify their influence, we define the local EMT-induced tracking-dynamic mismatch as
    \begin{equation}
    \label{eq:E_def}\mathbf{E}_{i}(s)\triangleq \mathbf{G}_{\mathrm{in},i}(s)
        -\mathbf{I}_{2},\, i\in\{1,\ldots,n\}.
    \end{equation}

Then, we introduce an uncertainty weight $\mathbf{W}_i(s)\in \mathcal{RH}_\infty^{2\times 2}$ to cover the EMT-induced mismatch $\mathbf{E}_i(s)$ entrywise. Denoting the $(p,q)$ entries of $\mathbf{E}_i(s)$ and $\mathbf{W}_i(s)$ by ${E}_{i,pq}$ and ${W}_{i,pq}$, respectively, the coverage condition is
\begin{equation}\label{eq:W_entrywise_bound}
    \bigl |{E}_{i,pq}(j\omega)\bigr|\;\le \;\bigl|{W}_{i,pq}(j\omega)\bigr|, \quad \forall p,q\in\{1,2\}.
\end{equation}

Since the envelope condition is imposed only on the magnitude response, 
each entry of the weight $\mathbf{W}_i(s)$ is selected as a stable, biproper, and minimum-phase transfer function,
with all entries nonzero.
This choice allows the EMT-induced mismatch to be normalized entrywise as
    \begin{equation}
        \Delta_{i,pq}(s) \triangleq W_{i,pq}(s)^{-1}E_{i,pq}(s), \, \|
        \Delta_{i,pq}\|_{\infty}\le 1, \forall p,q\in\{1,2\},\label{eq:Delta_entry_def}
    \end{equation}
    and to be collected into $\mathbf{\Delta}_{i}(s)=[\Delta_{i,pq}(s)]$. In other words, the reference-to-terminal mismatch can be written as
    \begin{equation}\label{equ: E=W delta}
        \mathbf{E}_{i}(s)=\mathbf{W}_{i}(s)\odot \boldsymbol{\Delta}_{i}(s).
    \end{equation}

    Consequently, the EMT tracking map can be modeled as
    \begin{equation}
        \mathbf{G}_{\mathrm{in},i}(s) = \mathbf{I}_{2}+ \mathbf{W}_{i}(s)\odot\mathbf{\Delta}
        _{i}(s). \label{eq:Gin_entrywise_uncertainty}
    \end{equation}

    In the next subsection, we stack the local EMT tracking maps
$\mathbf{G}_{\mathrm{in},i}$ across all converters to obtain a network-level
$\mathbf{M}$-$\mathbf{\Delta}$ uncertain interconnection (UI) for
robust-stability analysis.

    \subsection{Uncertain Interconnection Structure for Robustness Analysis}

    This subsection considers the impact of local EMT-induced uncertainty at system level 
    and reformulates the microgrid model as a feedback interconnection between EMT-induced uncertainty and the IILT system.

    The entrywise relation \eqref{eq:Gin_entrywise_uncertainty} indicates that EMT dynamics can be viewed as additive uncertainties \cite{doyle2013feedback} on  inner-loop maps, i.e.,
    \begin{equation}
        \mathbf{G}(s)=\mathbf{I}_{2n}+\mathbf{W}(s)\odot\mathbf{\Delta}(s), \label{eq:G_uncertain_network}
    \end{equation}
    where the block-diagonal structure of inner loops is $\mathbf{G}(s) = \mathrm{blkdiag}\left( \mathbf{G}_{\mathrm{in},1}(s),\ldots
        ,\mathbf{G}_{\mathrm{in},n}(s) \right),$ with uncertainty weight as $ \mathbf{W}(s) = \mathrm{blkdiag}\left( \mathbf{W}_{1}(s),\ldots,\mathbf{W}
        _{n}(s) \right)$, and $\mathbf{\Delta}(s) = \mathrm{blkdiag}\left( \mathbf{\Delta}_{1}(s),\ldots
        ,\mathbf{\Delta}_{n}(s) \right)$.

    Embedding $\mathbf{G}(s)$ into the interconnection yields the model with local EMT-induced uncertainty
    \begin{equation}
        \Delta\mathbf{x}= \mathbf{G}(s)\mathbf{H}(s) \left( \Delta\mathbf{s}^{\mathrm{ref}}
        - \mathbf{N}(s)\Delta\mathbf{x}\right). \label{eq:whole_uncertain_system}
    \end{equation}
    
    When $\mathbf{G}(s)=\mathbf{I}_{2n}$, \eqref{eq:whole_uncertain_system} reduces
    to the IILT model \eqref{eq: MET MIMO System}. Therefore, EMT impacts 
    are conservatively absorbed into the uncertainty weight $\mathbf{W}(s)$ and enter the system-level closed loop through the structured perturbation
    in \eqref{eq:G_uncertain_network}, while the outer-loop dynamics and network
    coupling remain represented by $\mathbf{H}(s)$ and
    $\mathbf{N}(s)$, respectively.

\begin{remark}
The block-diagonal EMT mismatch generally results in a dense system response through the coupled network $\mathbf{N}(s)$. The weights $\mathbf{W}_i(s)$ are constructed under different network conditions to cover changes in the local reference-to-terminal mismatch. Thus, the network dependence of the local mismatch is reflected in $\mathbf{W}_i(s)$, while its propagation among converters is still described explicitly by $\mathbf{N}(s)$.
\end{remark}

In Fig.~\ref{fig:mu analysis}(b), viewing the structure as a feedback interconnection between EMT uncertainty and the IILT closed-loop system, we can define the complementary sensitivity function $\mathbf{T}(s)$ to capture the response of the IILT system to EMT perturbations. 
        Denote the output of the network-level mismatch $\mathbf{E}(s)\triangleq\mathbf{G}(s)-\mathbf{I}_{2n}=\mathbf{W}(s)\odot\mathbf{\Delta}(s)$ as $\mathbf{d}=\Delta \mathbf{x}-\Delta\mathbf{x}^{\rm ref}$.
Then, for a system with \(\Delta\mathbf s^{\rm ref}=0\), the IILT loop gives
$\Delta \mathbf{x}^{\rm ref}
=
-\mathbf{H}(s)\mathbf{N}(s)
\left(\Delta \mathbf{x}^{\rm ref}+\mathbf{d}\right)$, or equivalently, 
\begin{equation}\label{eq: Complem-Sensi}
    \Delta \mathbf{x}^{\rm ref}
=
\mathbf{T}(s)\mathbf{d},
\quad
\mathbf{T}(s)\triangleq
-\left(\mathbf{I}_{2n}+\mathbf{H}(s)\mathbf{N}(s)\right)^{-1}
\mathbf{H}(s)\mathbf{N}(s),
\end{equation}
which governs how the IILT system responds to the EMT-induced mismatch.
In this formulation, the overall system can be viewed as a feedback interconnection between $\mathbf{W}(s)\odot\mathbf{\Delta}(s)$ and $\mathbf{T}(s)$, as presented in Fig.~\ref{fig:mu analysis}.

    \section{Measuring the Model Validity via a Robust-Stability Certificate}\label{sec: RobustStability}

    With the microgrid 
    reformulated as the IILT
    loop with structured EMT uncertainty, the remaining question is how
    to quantify whether the EMT uncertainty is small enough to remain non-destabilizing. 
    To this end, we first rewrite it into a
    lifted interconnection form for standard $\mu$-analysis. Then we derive the frequency-resolved interaction index and finally propose the robust-stability condition.

    \subsection{Lifted Uncertain Interconnection Structure}\label{sec: IV-A}

    In Section \ref{sec: EMTFeedback}, the EMT perturbation is injected into the system dynamics through \eqref{eq:G_uncertain_network}. 
    Although this Hadamard-product form \eqref{equ: E=W delta} reflects the physical
    origin of the EMT influence, a lifted representation is required to cast the problem
    into the standard robust-stability framework.

    To obtain an equivalent lifted representation \cite{zhou1996robust}, the local perturbation
    block \eqref{equ: E=W delta} is reformulated as 
    \begin{equation}\label{equ: lifted E_i}
        \mathbf{E}_{i}(s)=\bar{\mathbf{B}}\,\widetilde{\mathbf{W}}_{i}(s)\,\boldsymbol
        {\delta}_{i}(s)\,\bar{\mathbf{C}},
    \end{equation}
    where $\boldsymbol{\delta}_{i}(s)$ and $\widetilde{\mathbf{W}}_{i}(s)$ are defined in \eqref{equ: def delta_i} and \eqref{equ: def W_i}, respectively. 
    \begin{equation}\label{equ: def delta_i}
        \boldsymbol{\delta}_{i}(s)\triangleq \mathrm{diag}\!\big(\Delta_{i,11}(s)
        ,\Delta_{i,12}(s),\Delta_{i,21}(s),\Delta_{i,22}(s)\big),
    \end{equation}
    \begin{equation}\label{equ: def W_i}
        \widetilde{\mathbf{W}}_{i}(s)\triangleq \mathrm{diag}\!\big(W_{i,11}(s),W
        _{i,12}(s),W_{i,21}(s),W_{i,22}(s)\big).
    \end{equation}
    We define $\bar{\mathbf{B}}$ and $\bar{\mathbf{C}}$ as channel-selection matrices.
    To construct them, let $\boldsymbol{\Pi}_{i}\in\mathbb{R}^{2n\times 2}$ be the selection matrix satisfying $\Delta\mathbf x_i=\boldsymbol\Pi_i^\top\Delta\mathbf x$,
    and let $\mathbf{e}_{1}=[1\;0]^{\top}$, $\mathbf{e}_{2}=[0\;1]^{\top}$ be the
    basis vectors in $\mathbb{R}^{2}$. This yields
\begin{equation}
\bar{\mathbf C}^\top=
\begin{bmatrix}
\mathbf e_1\
\mathbf e_2\
\mathbf e_1\
\mathbf e_2
\end{bmatrix},
\qquad
\bar{\mathbf B}=
\begin{bmatrix}
\mathbf e_1 \ \mathbf e_1 \ \mathbf e_2\  \mathbf e_2
\end{bmatrix},
\end{equation}
    \begin{equation}
        \mathbf{C}_{i} = \bar{\mathbf{C}}\boldsymbol{\Pi}_{i}^{\top}, \ \mathbf{B}_{i} = \boldsymbol{\Pi}_i\bar{\mathbf{B}},
    \end{equation}
    where $\mathbf{C}_{i}$ performs channel duplication and $\mathbf{B}_{i}$ performs
    channel recombination. Then, the lifting matrices $\mathbf{B}$ and $\mathbf{C}$ for the overall system are constructed as
    \begin{equation}
        \mathbf{B}=
        \begin{bmatrix}
            \mathbf{B}_{1} & \cdots & \mathbf{B}_{n}
        \end{bmatrix}, \; \mathbf{C}=
        \begin{bmatrix}
            \mathbf{C}_{1}^{\top} & \cdots & \mathbf{C}_{n}^{\top}
        \end{bmatrix}^{\top}.
    \end{equation}

After weighting, recombination, nominal propagation, and duplication, the signal returned to the uncertainty input is
    \begin{equation}
        \mathbf{E}(s)=\mathbf{B}\,\widetilde{\mathcal{\mathbf{W}}}(s)\,\widetilde{\mathbf{\Delta}}_e(s)\,\mathbf{C},
    \end{equation}
    where 
    $\widetilde{\mathbf{\Delta}}_e(s)=\mathrm{blkdiag}\big(\boldsymbol{\delta}_{1}(s),\dots,\boldsymbol{\delta}_{n}(s)\big)$ and $\widetilde{\mathcal{\mathbf{W}}}(s)=\mathrm{blkdiag}\big(\widetilde{\mathbf{W}}_{1}(s),\dots,\widetilde{\mathbf{W}}_{n}(s)\big)$.
    This lifted form separates the uncertainty structure $\widetilde{\mathbf{\Delta}}_e(s)$ from the known lifted frequency-dependent weight $\widetilde{\mathcal{\mathbf{W}}}
    (s)$, making $\widetilde{\mathbf{\Delta}}_e(s)$ a diagonal matrix, and therefore enabling standard structured singular-value analysis.

    Then we derive the lifted representation of $\mathbf{T}(s)$ in \eqref{eq: Complem-Sensi}. 
    Under the representation
    above, the lifted nominal system matrix $\widetilde{\mathbf{T}}(s)$  becomes
    \begin{align}
        \label{eq:lifted_T}
        \widetilde{\mathbf{T}}(s)\triangleq \mathbf{C}\,\mathbf{T}(s)\,\mathbf{B},
    \end{align}
     and the associated weighted loop matrix, denoted as $\widetilde{\mathbf{M}}(s)$, is
\begin{equation}\label{eq:liftedloop}
        \widetilde{\mathbf{M}}(s)\triangleq \widetilde{\mathbf{T}}(s)\,\widetilde{\mathcal{\mathbf{W}}}(s).
\end{equation}

    Fig.~\ref{fig:Interconnection} illustrates this lifting procedure: the
    original entrywise EMT uncertainty
    $\mathbf{W}(s)\odot\boldsymbol{\Delta}(s)$ is converted into an equivalent
    structured interconnection.
    Through this procedure, the normalized uncertainty
    $\widetilde{\mathbf{\Delta}}_e(s)$ is diagonal and explicitly separated from the lifted weight
    $\widetilde{\mathcal{\mathbf{W}}}(s)$.

    \begin{figure}[h]
        \centering
        \includegraphics[width=0.45\textwidth]{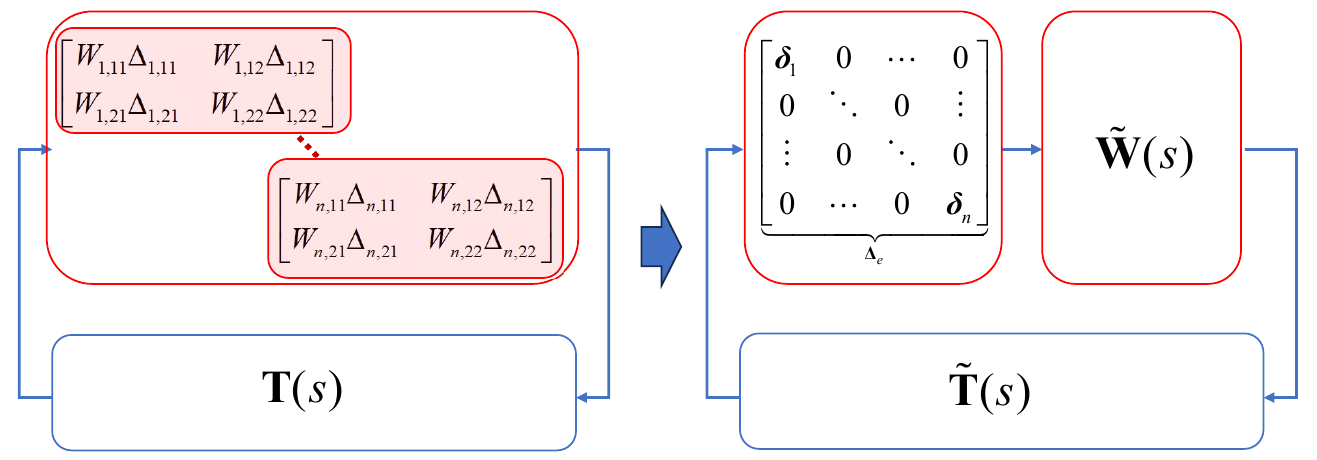}
        \caption{Lifting procedure converting the entrywise Hadamard-form EMT uncertainty into a standard structured interconnection compatible with $\mu$-analysis.}
        \label{fig:Interconnection}
    \end{figure}

    \subsection{EMT-IILT Interaction Index Induced by $\mu$-Analysis}\label{sec: IV-B}

    The lifted representation \eqref{eq:liftedloop} defines a frequency-resolved measure of the interaction between EMT-induced uncertainty and the IILT loop.
    
    \begin{definition}[Frequency-resolved interaction index]
        The frequency-resolved interaction index $\eta(\omega)$ is
    \begin{equation}
        \eta(\omega)\triangleq \mu_{\widetilde{\mathbf{\Delta}}_e}\!\big(\widetilde{\mathbf{M}}
        (j\omega)\big),
    \end{equation}
    where $\mu_{\widetilde{\mathbf{\Delta}}_e}(\cdot)$ denotes the
    structured singular value with respect to the lifted uncertainty structure $\widetilde{\mathbf{\Delta}}_e$.
    \end{definition}
    At each
    frequency $\omega$, it measures the loop gain of the feedback interconnection between the structured EMT uncertainty and
    the IILT loop.
    A larger value of $\eta(\omega)$ indicates a stronger interaction between converter-side EMT and IILT closed-loop dynamics, in the sense that the system is more susceptible to uncertainty-induced instability.

    \begin{definition}[Peak interaction value and critical frequency]
\label{def:peak_interaction}
Given the profile $\eta(\omega)$, the \emph{peak interaction value} is defined as
\begin{equation}
    \mu_{\max} \triangleq \sup_{\omega \in \mathbb{R}}\, \eta(\omega),
\end{equation}
and if this supremum is attained, the \emph{critical frequency} is defined as
\begin{equation}
    \omega^{\star} \triangleq \arg\max_{\omega \in \mathbb{R}}\, \eta(\omega).
\end{equation}
\end{definition}

In the subsequent analysis, the frequency-resolved index $\eta(\omega)$ 
identifies the spectral range of the strongest converter-side EMT--IILT interaction, 
while $\mu_{\max}$ serves as a scalar robustness margin quantifying the overall 
interaction severity.

    \subsection{Robust-Stability Certificate}\label{sec: IV-C}

    The peak structured singular value provides a sufficient certificate for
    robust stability under admissible EMT uncertainty, as established below.

    To derive the stability condition below, our analysis requires the following assumptions.

\begin{assumption}\label{asp:nominal}
(Nominal stability).
The IILT model is linearized at a stable equilibrium. Since the system
is invariant under a uniform angle shift, a projection onto the
relative-angle subspace (the quotient space) removes the common-angle
mode. In this reduced coordinate, the nominal interconnection is
internally stable, i.e.,
\[
(\mathbf{I}_{2n}+\mathbf{H}(s)\mathbf{N}(s))^{-1}
\mathbf{H}(s)\mathbf{N}(s)
\in \mathcal{RH}_{\infty}.
\]
\end{assumption}

\begin{assumption}\label{asp:wellposed}
    (Well-posed physical interconnection).
The assembled converter--network interconnection admits a unique causal
solution for all operating conditions considered in this paper.
Equivalently, the corresponding state-space or descriptor realization
is regular and impulse-free.
\end{assumption}

    \begin{assumption}\label{asp:coverage}
        (Envelope coverage over the operating range). The selected uncertainty
        weight $\mathbf{W}(s)$ provides a conservative envelope for all
        admissible EMT-induced mismatch over the entire frequency range and the
        operating-condition set.
    \end{assumption}

    Under these assumptions, the system can be viewed as a
    standard structured uncertain interconnection with the lifted weighted loop
    matrix $\widetilde{\mathbf{M}}(s)$ in \eqref{eq:liftedloop}.
    The following theorem gives a sufficient condition for robust stability.

    \begin{theorem}\label{thm:mu-certificate}
        Suppose Assumptions~\ref{asp:nominal}--\ref{asp:coverage} hold and the uncertain interconnection structure is well-posed
        for all admissible uncertainties.
        Then the linearized uncertain interconnection closed-loop system is robustly stable for every admissible
        structured uncertainty $\widetilde{\mathbf{\Delta}}_e(s)$ if
        \begin{equation}
            \sup_{\omega\in\mathbb{R}}\mu_{\widetilde{\mathbf{\Delta}}_e}\!\big
            (\widetilde{\mathbf{M}}(j\omega)\big)<1.\label{eq:CLmu_Thm1}
        \end{equation}
    \end{theorem}
    \begin{proof}
    Under the lifted representation \eqref{eq:liftedloop}, 
    the EMT-induced perturbation enters the IILT loop
    through the structured uncertainty block $\widetilde{\mathbf{\Delta}}_e
    (s)$ and the lifted weighted nominal operator $\widetilde{\mathbf{M}}(s)$. Robust stability is therefore equivalent to the non-singularity of
    \begin{equation}
        \det\!\big(\mathbf{I}_{4n}-\widetilde{\mathbf{\Delta}}_e(j\omega)\widetilde{\mathbf{M}}
        (j\omega)\big)\neq 0, \; \forall \omega\in\mathbb{R},
    \end{equation}
    for all admissible uncertainties. By the definition of the structured singular
    value \cite{skogestad2005multivariable}, the condition \eqref{eq:CLmu_Thm1}
    guarantees this nonsingularity at every frequency, and hence guarantees robust internal stability of the uncertain closed-loop system.
    \end{proof}

    The interaction index in Section \ref{sec: IV-B} thus provides a robust-stability condition.
    If its peak value satisfies $\mu_{\max}<1$, then the
    EMT-induced mismatch is non-destabilizing to the IILT model. Otherwise, robust stability cannot be guaranteed.     
    For implementation purposes, an alternative is given by
    $
        \|\widetilde{\mathbf{M}}(s)\|_{\infty}<1,
    $
    which is easier to evaluate but more conservative.

\begin{remark}[Implication for parameter design]
The certificate can also serve as an EMT-aware constraint for screening
phasor-domain parameter designs against prescribed EMT-model mismatch.
\end{remark}

\begin{remark}[Topology changes]
Topology changes require recomputing $\mathbf N(s)$, $\mathbf T(s)$,
and $\mu_{\max}$, while local weights $\mathbf W_i(s)$ may be reused
if their prescribed operating domains remain valid.
\end{remark}

\section{Implementation and Validation Approaches}\label{sec: PracticalImplementation}

Section~\ref{sec: RobustStability}  establishes the robust-stability certificate under a prescribed
EMT uncertainty weight. This section develops the practical procedure for
obtaining such a weight from either accessible EMT models or terminal
frequency-response data. The overall workflow is summarized in Fig.~\ref{fig:workflow},
from local uncertainty construction to system-level model-validity certification.

\begin{figure}[h]
\centering  \includegraphics[width=0.45\textwidth]{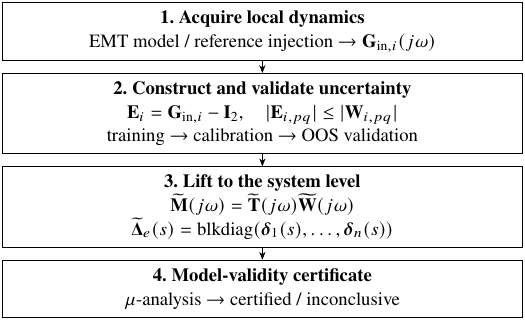}
\caption{Workflow from local uncertainty construction to system-level $\mu$-based model-validity certification.}
\label{fig:workflow}
\end{figure}

\subsection{Overall Uncertainty-Weight Construction Procedure}\label{sec:V-A}

For each converter $i$ and sampled system operating condition $\rho$, the local
reference-to-terminal response
$\mathbf{G}_{\mathrm{in},i,\rho}(j\omega)$ is first obtained using either the
model-based procedure in Section~\ref{sec: V-B} or the measurement-based procedure
in Section~\ref{sec: V-C}. The corresponding EMT-induced mismatch is
\begin{equation}
\mathbf{E}_{i,\rho}(j\omega)
=
\mathbf{G}_{\mathrm{in},i,\rho}(j\omega)-\mathbf{I}_2 .
\label{eq:local_mismatch_samples}
\end{equation}
Physically, $\mathbf{E}_{i,\rho}$ captures the terminal effect of EMT dynamics
neglected by the prescribed IILT model, while $\mathbf{W}_i(s)$ provides a
validated envelope of this effect.
The resulting mismatch samples are then used to construct, calibrate, and
independently validate a local uncertainty weight $\mathbf{W}_i(s)$ following the
common procedure in Section~\ref{Sec:V-D}. The model-based and measurement-based
routes therefore differ only in how
$\mathbf{G}_{\mathrm{in},i,\rho}(j\omega)$  is acquired; the subsequent uncertainty
construction and validation procedure is identical.

After validation, the local uncertainty weights are assembled as
\begin{equation}
\mathbf{W}(s)
=
\operatorname{blkdiag}
\bigl(
\mathbf{W}_1(s),\ldots,\mathbf{W}_n(s)
\bigr),
\label{eq:assembled_weight}
\end{equation}
which is incorporated into the lifted interconnection of Section~IV for
the $\mu$-based model-validity assessment.

\subsection{Model-Based Acquisition of EMT Mismatch}\label{sec: V-B}
    
This subsection presents a method for constructing the mismatch family $\{\mathbf{E}_{i,\rho}(s)\}_{\rho\in\mathcal{D}_{\rm op}}$ and then obtaining $\mathbf{W}_i(s)$ when a full model of the inner-loop and network dynamics is available.
For each operating point $\rho \in \mathcal{D}_{\rm op}$, 
with the ports defined in \eqref{eq: x_is_i}, the converter-side EMT dynamics are represented by the port map
\begin{equation}
    \Delta\mathbf{x}_i
    =
    \mathbf{\Phi}_{\mathbf{x},i,\rho}(s)\,\Delta\mathbf{x}_i^{\mathrm{ref}}
    +
    \mathbf{\Phi}_{\mathbf{s},i,\rho}(s)\,\Delta\mathbf{s}_i,
    \label{eq:device-port-model}
\end{equation}
where
$\mathbf{\Phi}_{\mathbf{x},i,\rho}(s)\in\mathbb{C}^{2\times 2}$ maps the local reference
perturbations to the terminal polar angle-voltage variables, and
$\mathbf{\Phi}_{\mathbf{s},i,\rho}(s)\in\mathbb{C}^{2\times 2}$ describes how the local
power-coupling channels affect the same terminal outputs.

For each sampled operating point $\rho$, the matrices
$\mathbf{\Phi}_{\mathbf{x},i,\rho}(s)$ and $\mathbf{\Phi}_{\mathbf{s},i,\rho}(s)$ are obtained by
linearizing the converter-side EMT subsystem around that operating condition.
The resulting model is written in the descriptor form
\begin{equation}
    s\mathbf{y}_i
    =
    \mathcal{A}_{i,\rho}\mathbf{y}_i
    +
    \mathcal{B}_{i,\rho}\mathbf{u}_i
    +
    \mathcal{E}_{i,\rho}\,s\mathbf{u}_i,
    \ 
    \Delta\mathbf{x}_i
    =
    \mathcal{C}_{i,\rho}\mathbf{y}_i
    +
    \mathcal{D}_{i,\rho}\mathbf{u}_i,
    \label{eq:descriptor-local-device}
\end{equation}
\begin{equation}
        \begin{bmatrix}\mathbf \Phi_{\mathbf x,i,\rho}(s),\mathbf \Phi_{\mathbf{s},i,\rho}(s)\end{bmatrix}
    =\mathcal C_{i,\rho}(s\mathbf{I}_8-\mathcal A_{i,\rho})^{-1}(\mathcal B_{i,\rho}+s\mathcal E_{i,\rho})+\mathcal D_{i,\rho},
\end{equation}
where $ \mathbf{u}_i \triangleq
    \begin{bmatrix}
        (\Delta\mathbf{x}_i^{\mathrm{ref}})^{\top}, 
        \Delta\mathbf{s}_i^{\top}
    \end{bmatrix}^{\top}$, and
$\mathcal{A}_{i,\rho}$,
$\mathcal{B}_{i,\rho}$,
$\mathcal{C}_{i,\rho}$,
$\mathcal{D}_{i,\rho}$,
$\mathcal{E}_{i,\rho}$
are the operating-point-dependent matrices produced by the local EMT
linearization.
In particular, the reference-angle-induced terms $s\Delta\theta_i^{\mathrm{ref}}$ are explicitly retained through the input-derivative matrix $\mathcal{E}_{i,\rho}$ and hence enter $\mathbf{G}_{\mathrm{in},i,\rho}(s)$; note that $\mathcal{E}_{i,\rho}$ is a model matrix of the linearized EMT subsystem, distinct from the mismatch transfer matrix $\mathbf{E}_i(s)=\mathbf{G}_{\mathrm{in},i}(s)-\mathbf{I}_2$.
The explicit construction of the transfer matrix from
\eqref{eq:descriptor-local-device}, together with the detailed symbolic
expressions, is deferred to Appendix~\ref{app:ABCDE}.

The local network dynamic response is represented through a local polar-to-power mapping
\begin{equation}
    \Delta\mathbf{s}_i
    =
    \mathbf{F}_{\mathrm{env},i,\rho}(s)\,\Delta\mathbf{x}_i,
    \label{eq:local-network-map}
\end{equation}
where
$\mathbf{F}_{\mathrm{env},i,\rho}(s)\in\mathbb{C}^{2\times2}$
denotes the local polar angle and voltage to active- and reactive-power response under operating condition $\rho$.
It is used only to close the local converter port when constructing
$\mathbf{G}_{\mathrm{in},i,\rho}$
and is not retained in the final interconnection; system-level network coupling remains represented by $\mathbf{N}(s)$.

Accordingly, the mismatch family can be constructed with the multi-GFM
network retained, while a single-converter infinite-bus (SCIB) Thevenin equivalent provides an
approximate surrogate when system-level sweeps are unavailable.

Substituting \eqref{eq:local-network-map} into
\eqref{eq:device-port-model} gives
\begin{equation}
    \Delta\mathbf{x}_i
    =
    \mathbf{\Phi}_{\mathbf{x},i,\rho}(s)\,\Delta\mathbf{x}_i^{\mathrm{ref}}
    +
    \mathbf{\Phi}_{\mathbf{s},i,\rho}(s)\,
    \mathbf{F}_{\mathrm{env},i,\rho}(s)\,
    \Delta\mathbf{x}_i.
\end{equation}
Therefore, the desired local polar interface map is obtained as
\begin{equation}
    \mathbf{G}_{\mathrm{in},i,\rho}(s)
    =
    \left(
    \mathbf{I}_2
    -
    \mathbf{\Phi}_{\mathbf{s},i,\rho}(s)\,
    \mathbf{F}_{\mathrm{env},i,\rho}(s)
    \right)^{-1}
    \mathbf{\Phi}_{\mathbf{x},i,\rho}(s),
    \label{eq:Gin-model-based}
\end{equation}
so that
$
    \Delta\mathbf{x}_i
    =
    \mathbf{G}_{\mathrm{in},i,\rho}(s)\,
    \Delta\mathbf{x}_i^{\mathrm{ref}}.
$

The construction in \eqref{eq:Gin-model-based} is carried out separately for
each sampled operating point $\rho \in \mathcal{D}_{\rm op}$, yielding a family of local
interface maps
$\{\mathbf{G}_{\mathrm{in},i,\rho}(s)\}_{\rho\in\mathcal{D}_{\rm op}}$. The resulting model-based mismatch samples are then used for the common uncertainty-weight construction and validation procedure in Section~\ref{Sec:V-D}.

    \subsection{Measurement-Based Acquisition of EMT Mismatch}\label{sec: V-C}

    When models of inner-loop dynamics are unavailable, proprietary, or only
    partially known, 
    $\mathbf{G}_{\mathrm{in},i}(s)$ and the EMT-induced uncertainty weight can be constructed from authorized reference-to-terminal frequency-response measurements.
    This section uses frequency-sweep data to construct the uncertainty weight.

    \subsubsection{Measurement Interface and Analysis Coordinates}

    For the $i$-th converter, the measurement interface is defined in polar
    variables: small-amplitude input signals are injected into the reference channels
    $\Delta\mathbf{x}^{\mathrm{ref}}_{i}=[\Delta\theta^{\mathrm{ref}}_{i},\Delta V^{\mathrm{ref}}_{i}]^{\top}$, and the resulting terminal response $\Delta\mathbf{x}_{i}=[\Delta
    \theta_{i},\Delta V_{i}]^{\top}$ is recorded, giving
    \begin{equation}
        \Delta\mathbf{x}_{i}=\mathbf{G}^{\mathrm{meas}}_{\mathrm{in},i}(s) \,\Delta\mathbf{x}
        ^{\mathrm{ref}}_{i}, \; \mathbf{E}^{\mathrm{meas}}_{i}(s)=\mathbf{G}^{\mathrm{meas}}_{\mathrm{in},i}(s)-\mathbf{I}_{2}.
    \end{equation}

    This identification is repeated across all operating conditions in $\mathcal{D}_{\rm op}$ to ensure that the resulting weight covers the prescribed uncertainty.

    \subsubsection{Identification Protocol}

    At frequency $\omega_{\mathrm{scan}}$, 
    small-amplitude sinusoidal perturbations are injected sequentially into the two reference channels $\Delta\theta_{i}^{\mathrm{ref}}$ and $\Delta V_{i}^{\mathrm{ref}}$. The responses of the two independent single-channel experiments are combined to form the two columns of $\mathbf{G}^{\mathrm{meas}}_{\mathrm{in},i}(j\omega_{\mathrm{scan}})$. 
    
    A full multiple-input multiple-output (MIMO) identification procedure is required because both the inner electromagnetic dynamics and the network dynamic feedback paths introduce cross-channel coupling between the angle and voltage channels.

    The perturbation may be superimposed on the normally generated outer-loop 
    reference using closed-loop identification, or applied with the outer-loop reference frozen at its operating-point 
    value. The operating point is selected from the prescribed family $\mathcal{D}_{\rm op}$, satisfies the steady-state power-flow equations, and is kept fixed during each frequency-response experiment.
    The injection amplitude must remain in a locally linear regime, outside dead zones and away from saturation and current-limiting boundaries, while maintaining an adequate signal-to-noise ratio. The sweep is repeated at a reduced amplitude, and agreement within a prescribed tolerance is used as a linearity check.

    In practice, frequency-response data are collected on a finite interaction
    band $[\omega_{\min},\omega_{\max}]$. The band is chosen to contain the dominant frequencies relevant to the robustness test. At low frequencies, $\mathbf{G}_{\mathrm{in}, i}(0)=\mathbf{I}_2$ makes the mismatch envelope sufficiently small near DC. At high frequencies, the roll-off of $\mathbf{T}(s)$ and the bounded high-frequency response of $\mathbf{G}_{\mathrm{in},i}(s)$ ensure that the weighted loop gain is negligible beyond this band.
    Thus, the measured band is sufficient if it covers the interaction
    frequencies of interest.

    The measurement-based construction requires an authorized reference-injection
    interface, such as an original equipment manufacturer (OEM)-provided diagnostic or a commissioning interface, but does not
    require disclosure of proprietary inner-loop models or controller parameters.

    \subsubsection{Acquisition of EMT Mismatch}

    For each condition $\rho$, we compute
    \begin{equation}
        \mathbf{E}^{\mathrm{meas}}_{i,\rho}(j\omega)= \mathbf{G}^{\mathrm{meas}}_{\mathrm{in},i,\rho}
        (j\omega)-\mathbf{I}_{2}.
    \end{equation}
    Subsequently, the resulting measurement-based mismatch samples are used for the common uncertainty-weight construction and validation procedure in Section~\ref{Sec:V-D}.

\subsection{Uncertainty-Weight Construction and Validation}\label{Sec:V-D}

Once the local mismatch responses $\mathbf{E}_{i,\rho}(j\omega)$ are obtained
from either the model-based or measurement-based procedure, the same protocol
is used to construct and validate the uncertainty weight. The sampled
operating points are divided a priori into disjoint training, calibration, and
test sets, denoted by $\mathcal{D}_{\rm tr}$, $\mathcal{D}_{\rm cal}$, and
$\mathcal{D}_{\rm test}$, respectively. Separate frequency grids $\Omega_{\rm tr}$, $\Omega_{\rm cal}$, and
$\Omega_{\rm test}$ are used for training, calibration, and independent testing,
respectively.

\subsubsection{Training and Calibration}

For each entry of the mismatch matrix, the training data are first used to
construct the pointwise envelope
\begin{equation}
e_{i,pq}^{\rm tr}(\omega)
=
\max_{\rho\in\mathcal{D}_{\rm tr}}
\left|E_{i,pq,\rho}(j\omega)\right|.
\label{eq:training_envelope}
\end{equation}
For each candidate order $m \in \mathcal{M}$ and inflation factor $a \in \mathcal{A}$, a stable, biproper, and minimum-phase rational candidate $W_{i,pq}^{(m,a)}(s)$ is fitted to the inflated training envelope $a \cdot e_{i,pq}^{\text{tr}}(\omega)$.

The candidate weights are subsequently evaluated on
$\mathcal{D}_{\rm cal}$. If coverage fails, we add training-frequency samples
within a prespecified log-frequency band around the largest calibration
violation and refit at fixed $(m,a)$. Among the candidates that cover all calibration samples, we select the one with the smallest $a$, breaking ties by the lowest order $m$. The final weight is then fixed as
\begin{equation}
\mathbf{W}_i(s)
=
\left[
W_{i,pq}(s)
\right],
\label{eq:final_weight}
\end{equation}
and is frozen before the independent test is performed.

\subsubsection{Independent Out-of-Sample Validation}
The frozen weight is evaluated at the untouched operating points in
$\mathcal{D}_{\rm test}$. For each channel, frequency, and operating point, we
define the coverage ratio
\begin{equation}
r_{i,pq}(\omega,\rho)
=
\frac{
\left|E_{i,pq,\rho}(j\omega)\right|
}{
\left|W_{i,pq}(j\omega)\right|
}.
\label{eq:coverage_ratio}
\end{equation}
The overall out-of-sample coverage is characterized by
\begin{equation}
r_{\max}
=
\max_i
\max_{p,q}
\max_{\rho\in\mathcal{D}_{\rm test}}
\max_{\omega\in\Omega_{\rm test}}
r_{i,pq}(\omega,\rho).
\label{eq:rmax_validation}
\end{equation}
Therefore, $r_{\max}<1$ indicates that the frozen uncertainty weights cover
all tested channels, frequencies, and operating points.

\begin{remark}
The uncertainty weight is intentionally conservative because it must cover
EMT mismatches over the prescribed operating range, with possible additional
slack from rational fitting. Calibration limits unnecessary overbounding,
while the out-of-sample test provides sample-based evidence of coverage.
\end{remark}

    \section{Case Study}\label{sec: Case}

    This section validates the proposed robustness framework on a simplified three-bus microgrid, a modified CIGRE microgrid with four GFM converters, and a modified IEEE
    33-bus microgrid with nine GFM converters.
    The well-posedness assumption is numerically checked for all operating
conditions considered in the case studies.

The three-bus system in Fig.~\ref{fig:3BusSystem} comprises three GFM converters interconnected through inductive tie lines with impedances $z_{13} = 0.3j+0.09~\text{p.u.}$ and $z_{12} = 0.6j+0.18~\text{p.u.}$, with a constant-impedance load of admittance $Y_{\mathrm{load}} = [0.0371-0.0011j,0.0371-0.0011j,0.0758-0.0022j]~\text{p.u.}$ at each bus and a dynamic network $\mathbf{N}(s)=\mathbf{G}_{\mathrm{net}}(s)$.
We set the operating point as $\mathbf{x}^* = [0, 1.0, 0.01919, 1.0, 0.0968, 1.0]^{\top}$.
The converter dynamics have been modeled in Section \ref{sec: MicrogridModel} and the parameters are presented in Table~\ref{tab:gfm_parameters}, where parameters without converter indices are identical across the three GFMs.
    \begin{figure}[h]
        \centering  \includegraphics[width=0.45\textwidth]{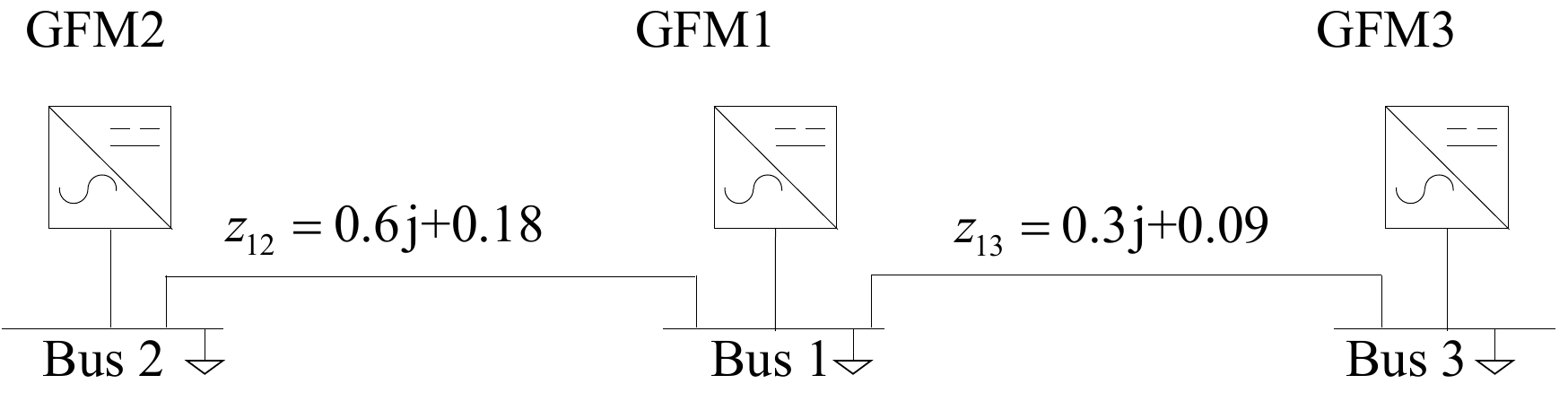}
        \caption{Three-bus GFM microgrid used for case study validation. Three GFM converters are interconnected through tie lines, each bus equipped with a constant-impedance load.}
        \label{fig:3BusSystem}
        \end{figure}

    \begin{table}[ht]
    \centering
    \caption{GFM outer-loop, inner-loop, and filter parameters}
    \label{tab:gfm_parameters}
    \begin{tabular}{c c c}
        \hline
        Category & Parameter & Value (p.u.) \\
        \hline
        \multirow{3}{*}{Outer loop}
        & $M$ & 1.5 \\
        & $D$ & 3.5 \\
        & $D_v$ & 0.001 \\
        \hline
        \multirow{4}{*}{Voltage/current loops}
        & $k_{\mathrm{vp}}$ & 0.2, 0.15, 0.25 \\
        & $k_{\mathrm{vi}}$ & 0.1 \\
        & $k_{\mathrm{ip}}$ & 1.5 \\
        & $k_{\mathrm{ii}}$ & 39.27 \\
        \hline
        \multirow{3}{*}{LC filter}
        & $L_\mathrm{f}$ & 0.044 \\
        & $C_\mathrm{f}$ & 0.012 \\
        & $r_\mathrm{f}$ & 0.001 \\
        \hline
    \end{tabular}
\end{table}

\subsection{Uncertainty-Weight Construction and Validation}

We first validate the uncertainty-weight construction of Section \ref{sec: PracticalImplementation}
under joint converter-operating-point and network variations.
For each converter $i\in\{1,2,3\}$, we evaluate the local mismatch
\begin{equation}
    \mathbf{E}_{i,\rho}(s)
    =
    \mathbf{G}_{\mathrm{in},i,\rho}(s)-\mathbf{I}_2,
\end{equation}
with the multi-GFM network retained.
The training set contains 81 joint operating conditions formed by
\[
\frac{P_2^{\rm ref}}{P_{2,0}^{\rm ref}},
\frac{P_3^{\rm ref}}{P_{3,0}^{\rm ref}}
\in[0.5,1.4],
\]
together with uniform network-impedance scales in
$[0.8,1.2]$ and $R/X\in[0.2,0.4]$. 
At each condition, the bus-1 slack injection and the nonlinear EMT
equilibrium are recomputed, while the reactive-power and voltage
references, controller parameters, and network topology remain fixed.
The mismatch responses are sampled at 72 frequencies over
$[1,10^{3.5}]~\mathrm{rad/s}$.

\begin{figure}
    \centering
    \includegraphics[width=0.8\linewidth]{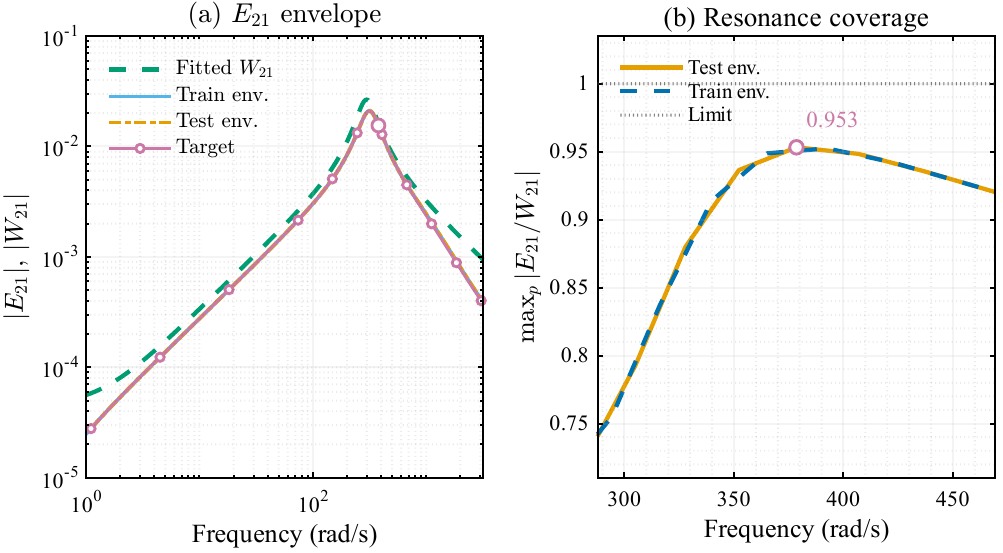}
    \caption{Uncertainty-weight validation for converter~1. (a) Fitted $W_{21}$ and sampled mismatch envelopes near the targeted resonance. (b) Training and untouched-test $W_{21}$ coverage ratios, both below the unit limit.}
\label{fig:weight_validation}
\end{figure}

For this case, $\mathcal{M}=\{2,4\}$ and $\mathcal{A}=\{1.05,1.10,1.20\}$,
and the procedure of Section~\ref{Sec:V-D} selects second-order local
weights with a $1.05$ envelope-inflation margin.
The frozen weights are then tested on operating conditions and frequencies
disjoint from training and calibration (Table~\ref{tab:weight_validation}).
Fig.~\ref{fig:weight_validation} shows the fitted $W_{21}$ against the training and untouched-test mismatch envelopes near the targeted resonance.
Over the 12 local mismatch channels, the largest test coverage ratio is
$r_{\max}^{\rm test}=0.9532<1$.
Thus, no violation is observed on the predeclared finite test grid.
These results provide finite-grid evidence that the frozen local weights
cover the sampled joint operating-condition family; they do not establish
continuous-frequency containment or EMT stability at every operating
point.

\begin{table}[htbp]
\centering
\caption{Sampling protocol and finite-grid coverage.}
\label{tab:weight_validation}
\begin{tabular}{lrrr}
\hline
Dataset          & $N_{\rho}$ & $N_{\omega}$ & $r_{\max}$ \\
\hline
Training         & 81 &   72 & 0.9524 \\
Calibration      & 18 &   72 & 0.9524 \\
Untouched test   & 18 &   71 & \textbf{0.9532} \\
Targeted corner  &  4 &   71 & 0.9529 \\
Peak OOS         &  8 & 1125 & 0.9256 \\
\hline
\end{tabular}
\end{table}

  \subsection{Interaction Index Validation}
    In the first case, we demonstrate that the phasor-domain model can yield an incorrect stability conclusion
    by constructing a case in which the IILT model with a dynamic network is stable while the EMT model is unstable, and we then validate the proposed frequency-resolved interaction index.

Three closed-loop systems are compared at this operating
point: the IILT model, the averaged EMT model, and an admissible UI realization
obtained with $\boldsymbol{\delta}_{i}(s)=-\mathbf{I}_4$ (the local normalized uncertainty block defined in \eqref{equ: def delta_i}). The last is a
counterexample within the prescribed uncertainty set, not a representation of
the actual EMT dynamics.
In this case, we use the model-based method to calculate the uncertainty weights $\mathbf{W}_i(s)$. The frozen uncertainty weights validated
in Section VI-A are used throughout this case.

\begin{figure}[h]
        \centering
        \includegraphics[width=0.48\textwidth]{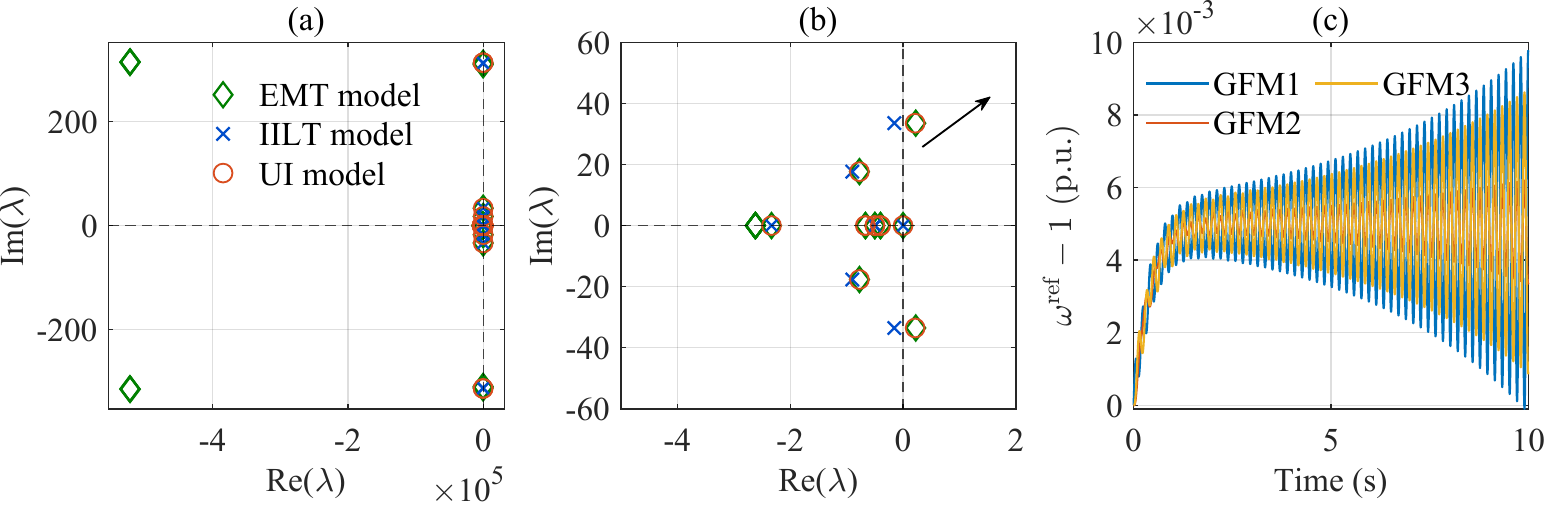}
        \caption{Validation scenario in which the IILT model predicts stability while the averaged EMT model is unstable. (a) Eigenvalues of the averaged EMT model, the IILT model, and an admissible uncertain-interconnection (UI) realization with $\boldsymbol{\delta}_{i}(s)=-\mathbf{I}_4$. (b) Zoomed view showing unstable poles of the averaged EMT model and the selected UI realization; the latter is an admissible counterexample and does not represent the actual EMT dynamics. (c) Time-domain EMT simulation of converter frequencies $\omega_i^{\mathrm{ref}}(t)$, extended to $10~\mathrm{s}$, exhibiting oscillations at approximately $33.5
        ~\mathrm{rad/s}$.}
        \label{fig:Poles+EMT}
        \end{figure}

    The poles of the three closed-loop systems are computed and compared in Fig.~\ref{fig:Poles+EMT}. The IILT model predicts that all modes on the electromechanical timescale are stable, with eigenvalues confined to the left-half plane. However, the averaged EMT model reveals that the same electromechanical modes migrate into the right-half plane due to destabilizing interactions introduced by the inner-loop dynamics, which are absent from the reduced-order model.
    The time-domain simulation of the averaged EMT model in Fig.~\ref{fig:Poles+EMT}(c) validates these predictions; the trajectories of $\omega_i^{\mathrm{ref}}(t)$ oscillate at approximately $\omega_{\mathrm{osc}} = 33.5~\mathrm{rad/s}$, with a continuously growing oscillation envelope over the extended $10~\mathrm{s}$ window, consistent with the right-half-plane eigenvalue pair in Fig.~\ref{fig:Poles+EMT}(b).

    \begin{figure}[h]
    \centering
    \includegraphics[width=0.48\textwidth]{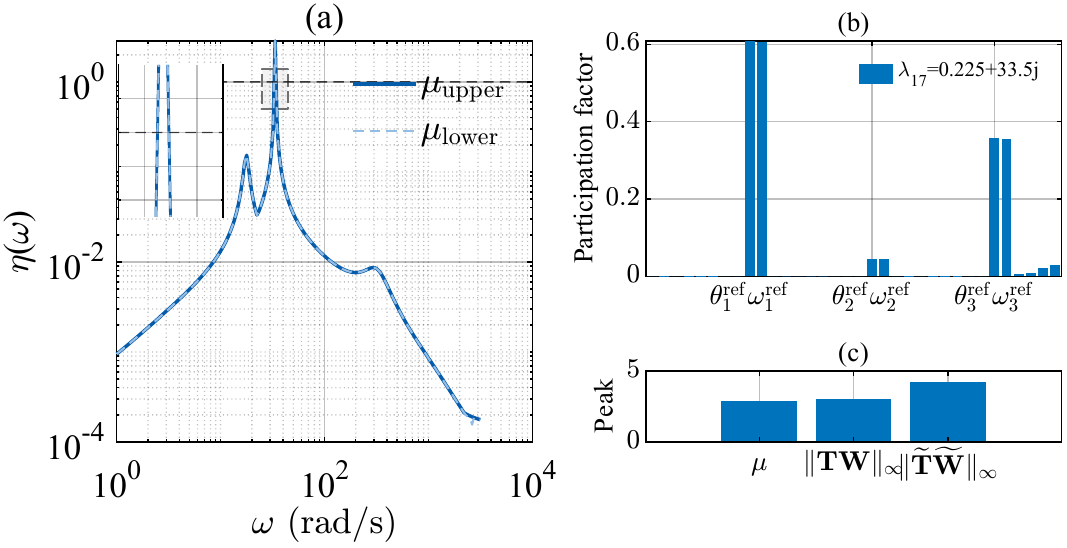}
    \caption{Frequency-resolved EMT--IILT interaction analysis for the unstable scenario. (a) Lower and upper numerical bounds on $\eta(\omega)$ versus frequency, with a peak near $33.5~\mathrm{rad/s}$. (b) Participation factors of the unstable mode. (c) Comparison of peak values among $\mu=\eta(\omega)$, $\|\mathbf{TW}\|_{\infty}$, and $\|\mathbf{\widetilde{T}\widetilde{W}}\|_{\infty}$.}
    \label{fig:Robustconditon}
    \end{figure}

    We next evaluate the proposed interaction index in Fig.~\ref{fig:Robustconditon}(a).
    The sampled \texttt{mussv} bounds show a lower bound above unity near $33.5~\mathrm{rad/s}$, so $\eta(\omega)>1$ at those frequencies. The peak near this frequency aligns with the dominant oscillation frequency $\omega_{\mathrm{osc}} = 33.5~\mathrm{rad/s}$ identified from the EMT simulation, 
    as annotated in Fig.~\ref{fig:Poles+EMT}(c).

    Furthermore, the participation analysis in Fig.~\ref{fig:Robustconditon}(b) reveals that the unstable mode is dominated by electromechanical states $\theta_i^{\mathrm{ref}},\omega_i^{\mathrm{ref}}$ across multiple converters. However, this instability is missed by phasor-domain models, and timescale separation cannot identify the issue, because the critical inner-loop states introduce phase lag and have small participation factors. In Fig.~\ref{fig:Robustconditon}(c),
    the comparison among $\mu,
    \|\mathbf{TW}\|_{\infty}$ and $\|\mathbf{\widetilde{T}\widetilde{W}}\|_{\infty}$ also shows that the structured $\mu$ criterion reduces conservatism. 

    \subsection{Robust-Stability Condition Validation}
    In this subsection, we use the same network and device model, varying operating points and outer-loop control parameters, to validate the proposed robust-stability condition.

    To evaluate the relationship between our proposed condition and the true EMT stability boundary, we conduct a parameter sweep of converter 3 in the range of $D_3\in[3,15]$ and $\alpha_3 \in [0.4, 2.0]$, where $\alpha_3$ scales the inner-loop gains as $(k_{\mathrm{vp},3},k_{\mathrm{vi},3},k_{\mathrm{ip},3},k_{\mathrm{ii},3}) = \alpha_3(0.1,0.1,1.5,39.27)$ and represents inner-loop bandwidth and EMT non-ideality. All other conditions are consistent with the default ones. 
For each $\alpha_3$, $\mathbf{W}_3^{(\alpha_3)}(s)$ is reconstructed
once and then held fixed over the corresponding $D_3$ sweep, while
$\mathbf{W}_1(s)$ and $\mathbf{W}_2(s)$ remain fixed throughout.

    \begin{figure}[h]
    \centering
    \includegraphics[width=0.48\textwidth]{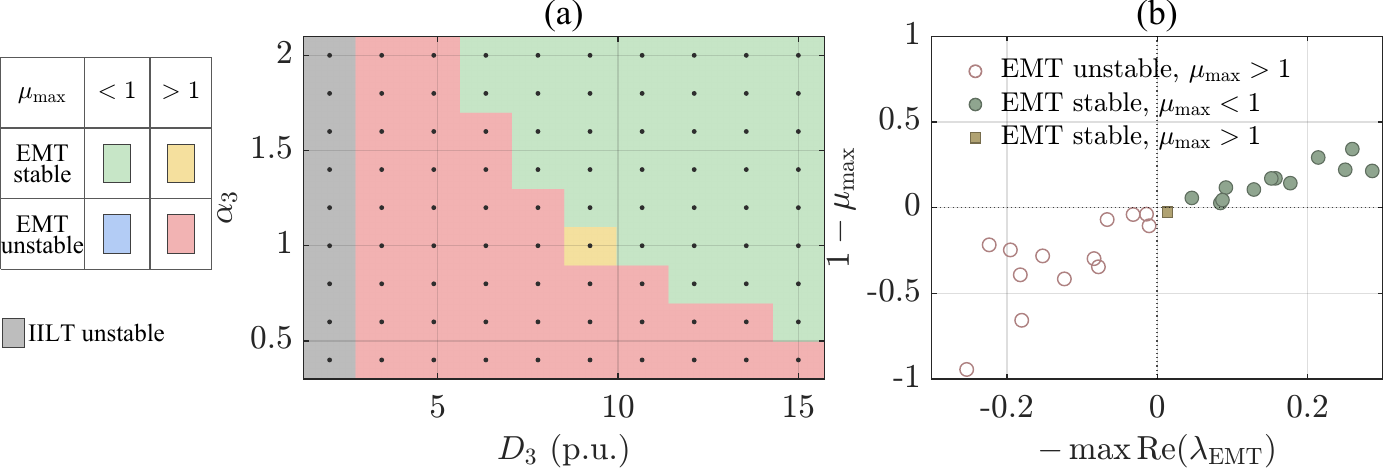}
    \caption{Multi-operating-point validation of the sufficient robust-stability certificate over the parameter space $(D_3,\alpha_3)$. (a) Classification of 90 parameter combinations into four nonempty regions among five logically possible categories. No sampled point is simultaneously IILT-stable, EMT-unstable, and certified by ${\mu}_{\max}<1$. (b) Scatter plot of $1-\mu_{\max}$ against the maximum real part of EMT eigenvalues after excluding the common-angle mode.}
    \label{fig:Multi-operating-point}
    \end{figure}

The sweep results are presented in Fig.~\ref{fig:Multi-operating-point}. No \texttt{mussv} bounds straddle unity in this sweep, so the labels
$\mu_{\max}<1$ and $\mu_{\max}>1$ are numerically unambiguous.
    In the sampled parameter grid, all IILT-stable operating points certified by ${\mu}_{\max}<1$ are EMT-stable. The stable-but-uncertified region (yellow in Fig.~\ref{fig:Multi-operating-point}) contains EMT-stable points with $\mu_{\max}>1$ and therefore illustrates the conservatism of the structured uncertainty bound.
For the representative stable-but-uncertified case at $(\alpha_3,D_3)=(1,9.2222)$, with rightmost non-gauge EMT eigenvalues $\lambda_{\mathrm{EMT}}=-0.0134\pm\mathrm{j}\,33.4174~\mathrm{s}^{-1}$, the lower/upper $\mu$ bounds are $1.0261799/1.0261852$, confirming that non-certification is due to the uncertainty set rather than numerical $\mu$ bounds.
    None of the 90 sampled points is simultaneously IILT-stable, EMT-unstable,
    and certified by $\mu_{\max}<1$. This observation is consistent with
    Theorem~\ref{thm:mu-certificate}.

    \subsection{Measurement-to-Certificate Procedure Validation}
    To evaluate the proposed measurement-to-certificate procedure, we consider a modified CIGRE medium-voltage distribution network with GFMs, shown in Fig.~\ref{fig:CIGRE}(a). Load parameters follow \cite{Barsali2014BenchmarkS}. GFMs are connected at buses 1, 5, 9, and 14; those at buses 5, 9, and 14 have an active-power output of $4~\text{MW}$ and a terminal voltage of $1~\text{p.u.}$ The converter parameters follow Table~\ref{tab:gfm_parameters}, except that $(M_i,D_i) = (3,14)~\text{p.u.}$ and $(k_{\mathrm{vp},i},k_{\mathrm{vi},i})=\alpha_i(0.1,0.1)$, with $\mathbf{\alpha}=(1, 0.8, 1.2, 1.5)$.

    \begin{figure}[h]
        \centering
        \includegraphics[width=0.48\textwidth]{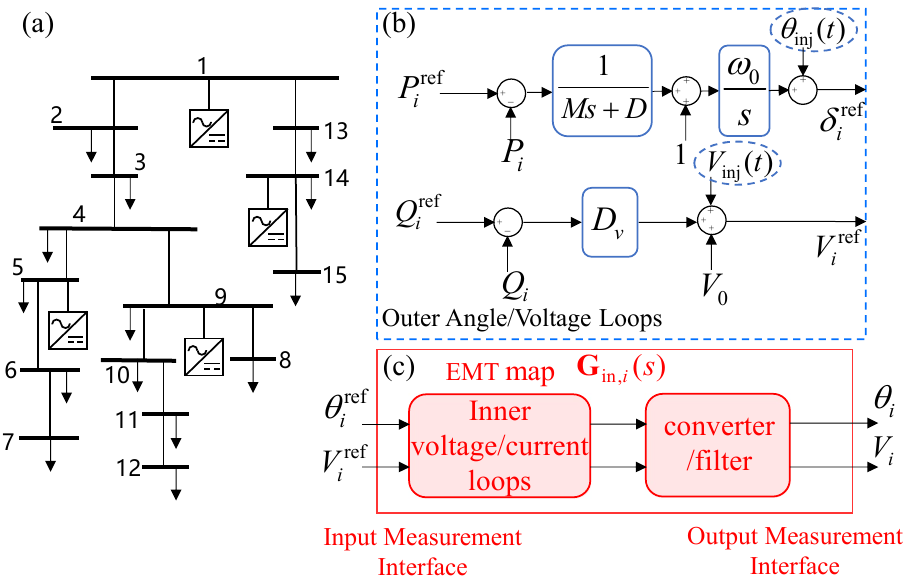}
        \caption{System configuration and converter-interface modeling. (a) CIGRE MV benchmark feeder with selected buses replaced by GFM converters. (b)-(c) Hierarchical representation of a representative GFM unit, including the outer angle/voltage loops and the converter-side EMT tracking dynamics from reference perturbations to terminal responses. The resulting terminal map $\mathbf{G}_{\mathrm{in},i}$ provides the basis for measurement-based uncertainty construction.}
        \label{fig:CIGRE}
        \end{figure}
    For each selected GFM unit, the converter-side terminal tracking dynamics are obtained in an SCIB measurement setup over the operating family $\mathcal{D}_{\rm op}$: $P^{\mathrm{ref}} \in [0, 2]~{\rm p.u.}$, $R/X \in [0, 1.2]$, and short-circuit ratio $\mathrm{SCR} \in [6, 15]$. The test-signal injection and input/output selections are shown in Fig.~\ref{fig:CIGRE}(b) and (c).

    To assess whether validity can be evaluated without access to the inner-loop model, we assume that the inner-loop dynamics of all converters are unavailable and that their parameters are heterogeneous. The local tracking map $\mathbf{G}_{\mathrm{in},i}(s)$ is identified via the measurement-based sweep procedure in Section \ref{sec: V-C}. Single-frequency sinusoidal signals are injected sequentially in the SCIB system with the outer loops frozen over $\omega_{\text{scan}} \in [1, 10^{3.5}]~\mathrm{rad/s}$. The angle and voltage injection amplitudes are $A_{\theta}=0.05~\mathrm{rad}$ and $A_V=0.05~\text{p.u.}$, respectively, with $\theta_{\text{inj}}(t)=A_{\theta}\sin(\omega_{\text{scan}}t)$ and $V_{\text{inj}}(t)=A_V\sin(\omega_{\text{scan}}t)$.
    
    \begin{figure}[h]
        \centering
        \includegraphics[width=0.9\columnwidth]{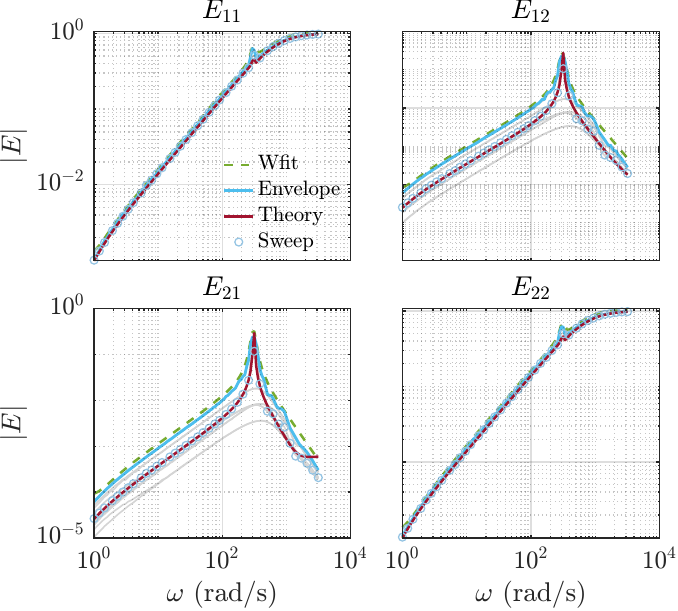}
        \caption{Measurement-based uncertainty construction for converter 1. The magnitudes of the four EMT-induced mismatch entries are compared among the analytical result (Theory), frequency-sweep identification (Sweep), pointwise envelope (Envelope), fitted uncertainty weight (Wfit $\mathbf{W}_1(s)$), and mismatches at other operating conditions (grey).}
        \label{fig:GB1_Gin_compare}
        \end{figure}
    \begin{figure}[h]
        \centering
        \includegraphics[width=\columnwidth]{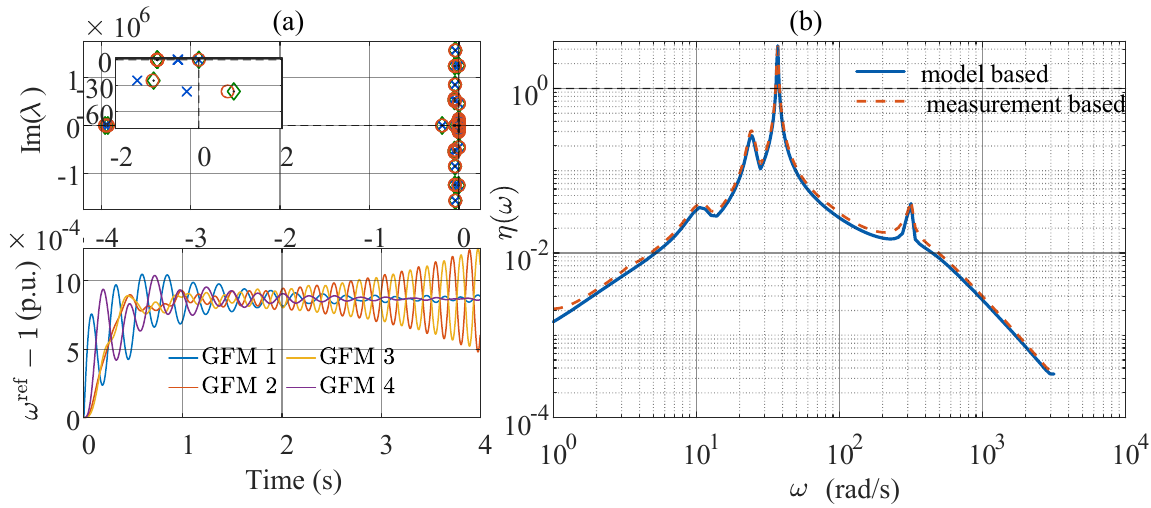}        \caption{(a) Eigenvalue distribution and time-domain frequency trajectories of the four GFM converters, showing that the IILT model is stable while the EMT model is unstable. (b) Comparison between model-based and measurement-based $\eta(\omega)$ results; the measurement-based curve closely tracks the model-based result across the analysis band and is marginally more conservative.}
        \label{fig:Model_Meas_compare}
        \end{figure}

    Using the SCIB setup for converter 1 (with frozen outer loops) as an example, Fig.~\ref{fig:GB1_Gin_compare} validates the measurement-based results in Section \ref{sec: PracticalImplementation} by comparing the magnitudes of the four EMT-induced mismatch entries between the analytical result and the frequency-sweep identification. The sweep-identified entries agree with the analytical benchmark across the analysis band, including the dominant magnitude variation and key crossover features, while the pointwise envelope and the fitted weight $\mathbf{W}_1(s)$ are seen to bound the identified mismatch.

    In Fig.~\ref{fig:Model_Meas_compare}, the interaction-index results show that, in this case, the measurement-based method closely matches the model-based method and is more conservative, supporting numerical consistency without using inner-loop model parameters.
    This case also provides empirical evidence that SCIB-based sweeps can be adequate for local uncertainty characterization in some settings.

\begin{table}[htbp]
\centering
\caption{Sensitivity to measurement settings. The benchmark is
16 cycles, 40~dB signal-to-noise ratio (SNR), and 0.01~ms output/data-acquisition (DAQ) delay, with
$\operatorname{mean}| E^{\rm meas}-E^{\rm theory}|=0.03824$.}
\label{tab:measurement_sensitivity}
\small
\begin{tabular}{lcc}
\toprule
Factor & Tested settings & $\Delta$ mean error (\%) \\
\midrule
Window (cycles)
& $16 \rightarrow 4 \rightarrow 2$
& $0 \rightarrow +0.75 \rightarrow +1.76$ \\

SNR (dB)
& $40 \rightarrow 30$
& $0 \rightarrow +0.78$ \\

Delay (ms)
& $0.01 \rightarrow 0.02 \rightarrow 0.05$
& $0 \rightarrow +1.23 \rightarrow +5.63$ \\
\bottomrule
\end{tabular}
\end{table}

Table~\ref{tab:measurement_sensitivity} summarizes the single-factor sensitivity to window length, measurement SNR, and output/DAQ delay. The identification remains relatively insensitive to these practical effects, with delay showing the largest impact. 
    
\subsection{Large-System Validation and Conservatism}

\begin{figure}[h]
\centering
\includegraphics[width=0.9\columnwidth]{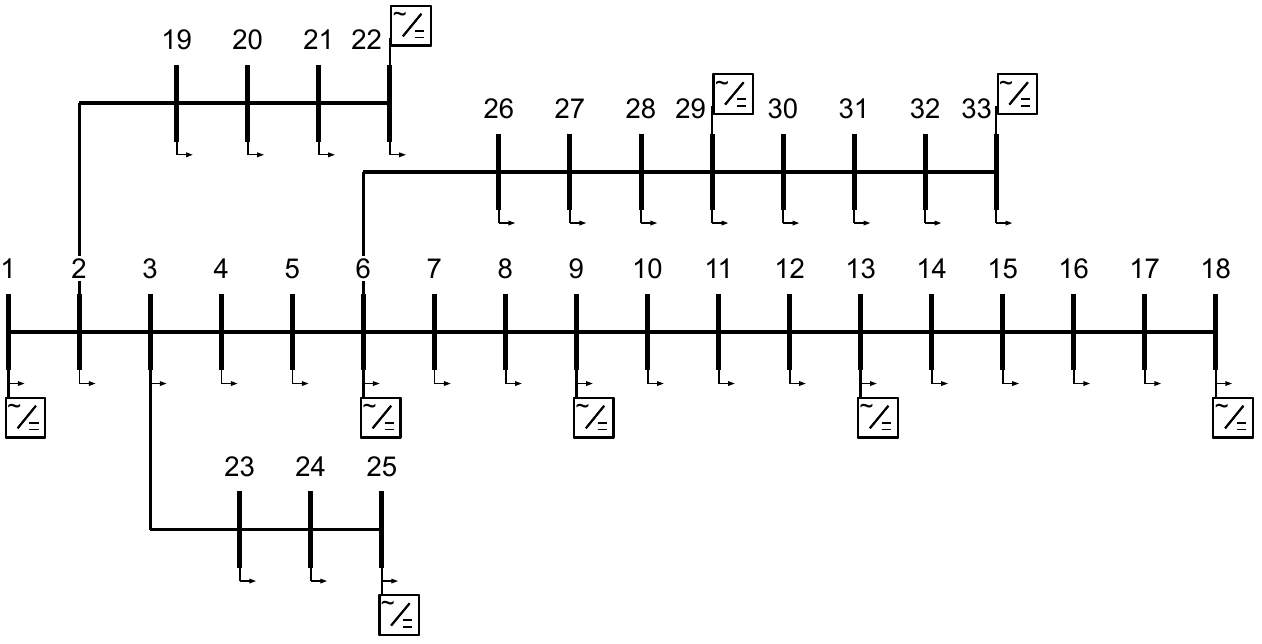}
\caption{Modified IEEE 33-bus system with 9 GFMs, including a droop-controlled converter at bus 9 and VSG-controlled units at the remaining GFM buses.}
\label{fig:IEEE33}
\end{figure}

\begin{figure}[h]
\centering
\includegraphics[width=1\columnwidth]{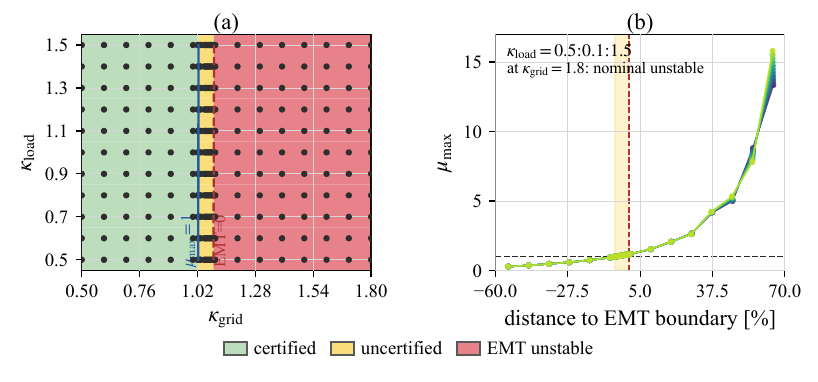}
\caption{Large-system multi-operating-point validation on the modified IEEE 33-bus nine-GFM system.
(a) Certification and EMT-stability regions over $(\kappa_{\mathrm{grid}},\kappa_{\mathrm{load}})$.
(b) Comparison between the $\mu_{\max}=1$ certification boundary and the EMT-stability boundary; the shaded region denotes the stable-but-uncertified margin.}
\label{fig:IEEE33scan}
\end{figure}

To assess the conservatism of the certificate in a larger system, we consider
a modified IEEE 33-bus network in Fig.~\ref{fig:IEEE33} with nine distributed GFMs, including one
droop-controlled unit and eight VSG-controlled units.
We sweep 308 operating conditions by varying the loading/dispatch factor
$\kappa_{\mathrm{load}}$ and the grid-strength factor
$\kappa_{\mathrm{grid}}$, with
$z_\ell=z_{\ell,0}/\kappa_{\mathrm{grid}}$.
The operating point and IILT/EMT linearizations are recomputed at each
condition, while the uncertainty weights remain fixed.

Fig.~\ref{fig:IEEE33scan} compares the certification and EMT-stability
boundaries over this operating range.
The certified region remains inside the EMT-stable region, with an average
boundary separation of approximately
$\Delta\kappa_{\mathrm{grid}}=0.069$ across the considered loading range.
Thus, the conservatism appears primarily as a narrow stable-but-uncertified region.

\section{Conclusion}\label{sec: Conclusion}

This paper examines the trustworthiness of phasor-domain GFM models for stability analysis in multi-converter microgrids. 
The central finding is that converter EMT-induced mismatch, introduced by non-ideal inner-loop and filter dynamics, can propagate through network coupling, invalidating stability conclusions drawn from phasor-domain models.
To certify model validity, the EMT-induced mismatch at each converter terminal is embedded as a structured uncertainty in the IILT feedback loop, yielding a frequency-resolved interaction index $\eta(\omega)$ and a $\mu$-based robust-stability certificate. The resulting condition $\mu_{\max}<1$ does not require strict timescale separation or inner-loop parameter disclosure, and admits both model-based and measurement-based implementations.

The case studies show operating points at which the IILT model predicts stability while the corresponding EMT model is unstable. The proposed measure
localizes the frequency range where EMT uncertainty interacts most strongly with IILT dynamics, consistent with modal analysis and simulation results. The structured certificate is less conservative than an unstructured small-gain bound. Measurement-based uncertainty weights allow the framework to use authorized reference-to-terminal frequency-response data for certification.

The certificate is local and sufficient: failure indicates loss of
certification rather than EMT instability, and its guarantee covers only
linear EMT effects represented by the reference-to-terminal uncertainty,
excluding mode-changing nonlinearities such as saturation and current
limiting. The measurement-based construction currently requires authorized
reference injection; future work will pursue non-invasive weight acquisition
for black-box models and online model-validity monitoring.

\appendices

\section{Descriptor Realization of the Local EMT Subsystem}
\label{app:ABCDE}
\begin{equation}
{\footnotesize
\begin{aligned}
\mathbf{y}_i \triangleq {}&
\bigl[
\Delta I_{\mathrm{L}d,i},\,
\Delta I_{\mathrm{L}q,i},\,
\Delta V_{\mathrm{od},i},\,
\Delta V_{\mathrm{oq},i}, \\
&\quad
\Delta X_{\mathrm{v}d,i},\,
\Delta X_{\mathrm{v}q,i},\,
\Delta X_{\mathrm{i}d,i},\,
\Delta X_{\mathrm{i}q,i}
\bigr]^{\top},
\end{aligned}
}
\label{eq:def_yi}
\end{equation}
{\footnotesize
where $\Delta X_{vd,i}, \Delta X_{vq,i}$ and $\Delta X_{id,i}, \Delta X_{iq,i}$ denote the integrator states of the inner voltage and current PI controllers respectively.}
\begin{equation}
{\footnotesize
\setlength{\arraycolsep}{2.5pt}
\renewcommand{\arraystretch}{0.95}
\mathcal{A}_{i,\rho}
=
\begin{bmatrix}
a_{11} & 0      & a_{13} & a_{14} & a_{15} & 0      & a_{17} & 0 \\
0      & a_{22} & a_{23} & a_{24} & 0      & a_{26} & 0      & a_{28} \\
a_{31} & 0      & a_{33} & a_{34} & 0      & 0      & 0      & 0 \\
0      & a_{42} & a_{43} & a_{44} & 0      & 0      & 0      & 0 \\
0      & 0      & -1     & 0      & 0      & 0      & 0      & 0 \\
0      & 0      & 0      & -1     & 0      & 0      & 0      & 0 \\
-1     & 0      & a_{73} & a_{74} & a_{75} & 0      & 0      & 0 \\
0      & -1     & a_{83} & a_{84} & 0      & a_{86} & 0      & 0
\end{bmatrix},
}
\label{eq:A_compact}
\end{equation}
{\footnotesize
\begin{align}
a_{11} &= a_{22} = -\tfrac{k_{\mathrm{ip},i}+r_{\mathrm{f},i}}{L_{\mathrm{f},i}}, &
a_{15} &= a_{26} = \tfrac{k_{\mathrm{ip},i}k_{\mathrm{vi},i}}{L_{\mathrm{f},i}}, &
\\
a_{17} &= a_{28} = \tfrac{k_{\mathrm{ii},i}}{L_{\mathrm{f},i}},&
a_{34} &= \omega_0+\tfrac{I_{\mathrm{o}q,i}^{*}}{C_{\mathrm{f},i}V_{\mathrm{o}d,i}^{*}},&
\\
a_{13} &= -\tfrac{k_{\mathrm{ip},i}\!\left(V_{\mathrm{o}d,i}^{*}k_{\mathrm{vp},i}+I_{\mathrm{o}d,i}^{*}\right)}
{L_{\mathrm{f},i}V_{\mathrm{o}d,i}^{*}}, &
a_{14} &= -\tfrac{k_{\mathrm{ip},i}\!\left(C_{\mathrm{f},i}V_{\mathrm{o}d,i}^{*}\omega_0+I_{\mathrm{o}q,i}^{*}\right)}
{L_{\mathrm{f},i}V_{\mathrm{o}d,i}^{*}},&
\\
a_{23} &= \tfrac{k_{\mathrm{ip},i}\!\left(C_{\mathrm{f},i}V_{\mathrm{o}d,i}^{*}\omega_0-I_{\mathrm{o}q,i}^{*}\right)}
{L_{\mathrm{f},i}V_{\mathrm{o}d,i}^{*}}, &
a_{24} &= \tfrac{k_{\mathrm{ip},i}\!\left(-V_{\mathrm{o}d,i}^{*}k_{\mathrm{vp},i}+I_{\mathrm{o}d,i}^{*}\right)}
{L_{\mathrm{f},i}V_{\mathrm{o}d,i}^{*}},&
\\
a_{31} &= a_{42} = \tfrac{1}{C_{\mathrm{f},i}}, &
a_{33} &= -a_{44} =  \tfrac{I_{\mathrm{o}d,i}^{*}}{C_{\mathrm{f},i}V_{\mathrm{o}d,i}^{*}}, &
\\
a_{43} &= -\omega_0+\tfrac{I_{\mathrm{o}q,i}^{*}}{C_{\mathrm{f},i}V_{\mathrm{o}d,i}^{*}}, & a_{75} &=a_{86} = k_{\mathrm{vi},i},  &
\\
a_{73} &= -k_{\mathrm{vp},i}-\tfrac{I_{\mathrm{o}d,i}^{*}}{V_{\mathrm{o}d,i}^{*}}, &
a_{74} &= -C_{\mathrm{f},i}\omega_0-\tfrac{I_{\mathrm{o}q,i}^{*}}{V_{\mathrm{o}d,i}^{*}}, &
\\
a_{83} &= C_{\mathrm{f},i}\omega_0-\tfrac{I_{\mathrm{o}q,i}^{*}}{V_{\mathrm{o}d,i}^{*}}, &
a_{84} &= -k_{\mathrm{vp},i}+\tfrac{I_{\mathrm{o}d,i}^{*}}{V_{\mathrm{o}d,i}^{*}}, &
\end{align}
}

\begin{equation}
{\footnotesize
\setlength{\arraycolsep}{2.2pt}
\renewcommand{\arraystretch}{0.92}
\mathcal{B}_{i,\rho}
=
\begin{bmatrix}
0 & b_{12} & b_{13} & 0 \\
0 & 0      & 0      & b_{24} \\
0 & 0      & b_{33} & 0 \\
0 & 0      & 0      & b_{44} \\
0 & 1      & 0      & 0 \\
0 & 0      & 0      & 0 \\
0 & b_{72} & b_{73} & 0 \\
0 & 0      & 0      & b_{84}
\end{bmatrix},
\quad
\mathcal{E}_{i,\rho}
=
\begin{bmatrix}
I_{\mathrm{L}q,i}^*  & 0 & 0 & 0 \\
-I_{\mathrm{L}d,i}^* & 0 & 0 & 0 \\
0                   & 0 & 0 & 0 \\
-V_{\mathrm{o}d,i}^{*} & 0 & 0 & 0 \\
0                   & 0 & 0 & 0 \\
0                   & 0 & 0 & 0 \\
0                   & 0 & 0 & 0 \\
0                   & 0 & 0 & 0
\end{bmatrix},
}
\label{eq:BE_matrices}
\end{equation}

{\footnotesize
\begin{align}
&b_{12}=\frac{k_{\mathrm{ip},i}k_{\mathrm{vp},i}}{L_{\mathrm{f},i}},\quad
b_{13}=-b_{24}=\frac{2k_{\mathrm{ip},i}}{3L_{\mathrm{f},i}V_{\mathrm{o}d,i}^{*}},\quad b_{72}=k_{\mathrm{vp},i},
\\
& b_{33}=-b_{44}=-\frac{2}{3C_{\mathrm{f},i}V_{\mathrm{o}d,i}^{*}} \notag,\quad
b_{73}=-b_{84}=\frac{2}{3V_{\mathrm{o}d,i}^{*}}.
\end{align}

\begin{equation}
\Delta\mathbf{x}_i
=
\mathcal{C}_{i,\rho}\mathbf{y}_i+\mathcal{D}_{i,\rho}\mathbf{u}_i ,
\label{eq:CD_def}
\end{equation}
where the only nonzero entries of
$\mathcal{C}_{i,\rho}$ and $\mathcal{D}_{i,\rho}$ are
\begin{equation}
{\footnotesize
\left[\mathcal{C}_{i,\rho}\right]_{1,4}
=
\tfrac{1}{V_{\mathrm{o}d,i}^{*}},
\qquad
\left[\mathcal{C}_{i,\rho}\right]_{2,3}
=1,
\qquad
\left[\mathcal{D}_{i,\rho}\right]_{1,1}=1 .
}
\label{eq:CD_nonzero}
\end{equation}
}

\section{Fitted Local Uncertainty Weights}
\label{app:fitted_weights}

The second-order local weights for Device~1 used in the case study are
reported below. Each weight is written as
\begin{equation}
    W_{pq}(s)=\frac{N_{pq}(s)}{D_{pq}(s)}.
\end{equation}
The local weights of the remaining converters are obtained by the same
procedure.

{\scriptsize

\begin{equation}
\begin{aligned}
N_{11}(s)={}&
1.01746s^2
+18.5415s
+3.60818,\\
D_{11}(s)={}&
s^2
+1.14924{\times}10^{3}s
+2.25602{\times}10^{4}.
\end{aligned}
\end{equation}

\begin{equation}
\begin{aligned}
N_{12}(s)={}&
4.02965{\times}10^{-5}s^2
+2.93829s
+3.97584,\\
D_{12}(s)={}&
s^2
+109.590s
+8.90168{\times}10^{4}.
\end{aligned}
\end{equation}

\begin{equation}
\begin{aligned}
N_{21}(s)={}&
4.48707{\times}10^{-5}s^2
+2.95853s
+3.99714,\\
D_{21}(s)={}&
s^2
+111.001s
+8.87408{\times}10^{4}.
\end{aligned}
\end{equation}

\begin{equation}
\begin{aligned}
N_{22}(s)={}&
1.01801s^2
+18.5751s
+3.65798,\\
D_{22}(s)={}&
s^2
+1.15464{\times}10^{3}s
+2.26812{\times}10^{4}.
\end{aligned}
\end{equation}

}

\section{Computational Complexity}

For an $n$-converter system evaluated at $N_{\omega}$
    frequency points, the approximate computational workload is
    \[
    \mathcal{C}_{\mathrm{cert}}\approx N_{\omega}\left[
    \frac{8}{3}(2n)^3+\kappa_{\mu}d_{\Delta}^{3}\right],
    \]
    where the first term corresponds to the dense construction of
    $\mathbf{T}(j\omega)$ and $d_{\Delta}$ is the dimension of the matrix
    supplied to \texttt{mussv}. $\kappa_{\mu}$ denotes an algorithm-dependent factor accounting for the computational overhead of numerical \(\mu\)-bound estimation.

For the uncertainty structure used here, $d_{\Delta}=4n$.
Thus, the present dense implementation scales linearly with $N_\omega$
and approximately cubically with $n$, apart from the block-dependent
\texttt{mussv} iterations. For the nine-GFM case in Section~VI-E, 201
frequency points over $[1,3162.28]~\mathrm{rad/s}$ require
$8.96\pm0.81$~s per operating point (median $8.85$~s, 95th percentile
$10.49$~s, over 108 timed runs) using MATLAB R2025a Update 1 on an Intel Core
i9-14900HX with 32~GB RAM.
        
    \bibliographystyle{IEEEtran}
    \bibliography{bibtex/bib/IEEEabrv, Citations1}

    \ifCLASSOPTIONcaptionsoff
    \newpage
    \fi

\end{document}